\documentclass[tikz, a4paper, 12pt, oneside]{book}
\usepackage[left=2.5cm,right=2.5cm,top=2.5cm,bottom=2.5cm]{geometry}
\usepackage[spanish, english]{babel}

\usepackage[latin1]{inputenc}
\usepackage{graphicx}
\usepackage{amssymb}
\usepackage{amsthm}
\usepackage{amsmath}
\usepackage{makeidx}
\usepackage{color}
\usepackage{mathtools}
\usepackage{braket}
\usepackage{tikz}
\usepackage{float}
\usepackage{listings}
\usepackage{hyperref}
\usetikzlibrary{arrows.meta}

\usepackage{circuitikz}
\usetikzlibrary{shapes, arrows.meta, positioning, patterns}
\numberwithin{figure}{chapter}

\tikzset{
  sg/.style={draw, minimum width=1.2cm, minimum height=1.2cm, align=center},
  absorber/.style={draw, pattern=north east lines, minimum height=0.4cm, minimum width=1.2cm},
  measurement/.style={circle, draw, minimum size=0.4cm, inner sep=0pt},
  myarrow/.style={-{Latex}, thick}
}

\lstdefinestyle{custompython}{
    language=Python,
    backgroundcolor=\color{white},   
    commentstyle=\color{gray},
    keywordstyle=\color{blue}\bfseries,
    numberstyle=\tiny\color{gray},
    stringstyle=\color{green!60!black},
    basicstyle=\ttfamily\footnotesize,
    breaklines=true,
    captionpos=b,
    numbers=left,
    numbersep=5pt,
    showspaces=false,
    showstringspaces=false,
    frame=single,
    rulecolor=\color{black},
}

\newtheorem{definition}{Definition}
\newtheorem{lemma}{Lemma}
\newtheorem{proposition}{Proposition}
\newtheorem{theorem}{Theorem}
\newtheorem{corollary}{Corollary}
\newtheorem{observation}{Observation}
\newtheorem{example}{Example}

\newtheorem*{problem}{Problem}
\newtheorem{postulate}{Postulate}
\newtheorem{remark}{Remark}
\newtheorem*{postulateprime}{Postulate 3\ensuremath{'}}

\begin{document}
\pagenumbering{roman}

\title{A Rigorous Introduction to Hamiltonian Simulation via High-Order Product Formulas}
\author{Javier Lopez-Cerezo\\
Department of Applied Mathematics \\
University of Malaga} 
\maketitle
\newpage
\mbox{}
\thispagestyle{empty}

\tableofcontents

\pagebreak

\chapter*{Abstract}
\addcontentsline{toc}{chapter}{Abstract} 
This work provides a rigorous and self-contained introduction to numerical methods for Hamiltonian simulation in quantum computing, with a focus on high-order product formulas for efficiently approximating the time evolution of quantum systems. Aimed at students and researchers seeking a clear mathematical treatment, the study begins with the foundational principles of quantum mechanics and quantum computation before presenting the Lie-Trotter product formula and its higher-order generalizations. In particular, Suzuki's recursive method is explored to achieve improved error scaling. Through theoretical analysis and illustrative examples, the advantages and limitations of these techniques are discussed, with an emphasis on their application to $k$-local Hamiltonians and their role in overcoming classical computational bottlenecks. The work concludes with a brief overview of current advances and open challenges in Hamiltonian simulation.

\pagebreak

\chapter*{Acknowledgements}
\addcontentsline{toc}{chapter}{Acknowledgements} 
This work stems from my Mathematics Bachelor's thesis, defended in June 2025 at the University of Malaga. I would like to express my sincere gratitude to my advisor, Carlos Pares, who welcomed such a niche and quantum-focused topic with curiosity and enthusiasm, and whose openness to exploring new topics made this project possible. His guidance, encouragement, and willingness to learn alongside me were invaluable throughout this process.

\pagebreak
 
\chapter*{Introduction}
\addcontentsline{toc}{chapter}{Introduction} 

One of the most impactful applications of computation lies in the simulation of physical systems. Such simulations are fundamental across numerous scientific and engineering disciplines, ranging from fluid dynamics and particle physics to cutting-edge fields like drug discovery and the design of novel materials. At the heart of these simulations is the solution of differential equations that mathematically encode the physical laws governing the dynamics of these systems.
\\\\
For example, Newton's second law of motion,
\begin{equation*}
m \frac{d^2 \mathbf{x}(t)}{dt^2} = \mathbf{F}\left(\mathbf{x}(t), \frac{d \mathbf{x}(t)}{dt}, t\right),
\end{equation*}
describes the trajectory of a particle of mass \(m\) subject to a force \(\mathbf{F}\). Similarly, the heat equation governs the diffusion of temperature in a medium:
\begin{equation*}
\frac{\partial u(\mathbf{x}, t)}{\partial t} = \alpha \nabla^2 u(\mathbf{x}, t) + Q(\mathbf{x}, t),
\end{equation*}
where \(u(\mathbf{x}, t)\) denotes the temperature field, \(\alpha\) the thermal diffusivity, \(Q(\mathbf{x}, t)\) internal heat sources or sinks and $\nabla^2u$ is the Laplacian of $u$.
\\\\
The primary goal in simulation is to determine the state of a system at a given time and/or position, starting from known initial conditions. To achieve this, the system's state is discretized in space and time and the governing differential equations are approximated using numerical methods. Through iterative computational procedures, the system evolves from its initial to its final state. Crucially, the numerical error introduced during this process must be carefully controlled and bounded. However, not all dynamical systems are equally amenable to classical simulation.
\\\\
When it comes to quantum systems, the situation becomes particularly challenging. The time evolution of many quantum systems is governed by the Schr\"{o}dinger equation,
\begin{equation*}
i\hbar \frac{d}{dt} |\psi(t)\rangle = H(t) |\psi(t)\rangle,
\end{equation*}
where \(|\psi(t)\rangle\) represents the quantum state at time $t$, $\hbar$ is Planck's constant  and \(H\) is the Hamiltonian operator associated with the system. In 1982, Richard Feynman \cite{feynman1982simulating} famously observed that simulating the full quantum dynamics of arbitrary systems on a classical computer quickly becomes intractable, as the dimension of the state space grows exponentially with the system size. For instance, a system of 40 spin-\(\frac{1}{2}\) particles requires storing \(2^{40} \approx 10^{12}\) complex numbers and simulating its time evolution involves manipulating matrices of size \(2^{40} \times 2^{40}\), which is beyond current classical computational capabilities.
\\\\
To overcome this exponential bottleneck, Feynman proposed the revolutionary idea of using one quantum system to simulate another directly, exploiting the fact that the simulator would evolve according to the same fundamental equations as those of the system being simulated. He conjectured the existence of universal quantum simulators capable of efficiently simulating any quantum system governed by local interactions. This conjecture was proven by Seth Lloyd in 1996 \cite{lloyd1996universal}, who demonstrated that quantum systems with local Hamiltonians can be efficiently simulated on a quantum computer using product formula techniques.
\\\\
In this work, we explore the mathematical foundations of quantum mechanics that describe quantum systems and investigate how the ideas introduced by Feynman and formalized by Lloyd form the basis for quantum computers as efficient simulators of quantum dynamics. Our goal is to provide an introduction to the understanding of how quantum computation has the potential to revolutionize the simulation of complex quantum systems, with significant implications for quantum chemistry, materials science and related fields.
\\\\
This work is organized as follows. Chapter 1 collects the essential linear algebra results required for the subsequent discussion. We assume familiarity with undergraduate-level linear algebra and thus omit proofs of standard results; references include the open-source text \cite{hefferon2021linear} and \cite{lax2013linear} for the specific results used here. Chapter 2 introduces the postulates of quantum mechanics and provides a high-level overview of quantum computation and its application to quantum simulation. Chapter 3 formalizes the topic of Hamiltonian simulation, presenting detailed treatments of the Lie-Trotter product formula, Suzuki's higher-order product formulas and includes a numerical validation of their properties implemented in Python. The corresponding code is included in the Appendix. Finally, Chapter 4 briefly discusses the state of the art and open questions in Hamiltonian simulation, highlighting current research challenges and directions.

\pagebreak
\chapter{Preliminaries}
\setcounter{page}{1}
\pagenumbering{arabic}

In this chapter, we adopt the standard notation of quantum mechanics to express linear algebraic concepts. In particular, we use Dirac's bra-ket notation: vectors are written as kets $\ket{\psi}$, their conjugate transposes as bras $\bra{\phi}$ and inner products as brackets $\braket{\phi|\psi}$.

\begin{definition}
A complex inner product space $V$ is:

\begin{enumerate}
    \item A complex vector space, that is,
    \[
    \ket{\psi},\ket{\varphi}\in V \text{ and } a,b\in\mathbb{C} \Rightarrow a\ket{\psi}+b\ket{\varphi} \in V,
    \]
    and the operations satisfy the standard vector space axioms.
    \item Equipped with an inner product
    \[
    (\cdot,\cdot):V \times V \longrightarrow \mathbb{C},
    \]
    \[
    (\ket{\psi},\ket{\varphi}) \longmapsto \langle\psi|\varphi\rangle,
    \]
    such that for all $\ket{\varphi}, \ket{\psi}, \ket{\varphi_1}, \ket{\varphi_2} \in V$ and $a, b \in \mathbb{C}$:
    \begin{align*}
    \langle\psi|\varphi\rangle &= \langle\varphi|\psi\rangle^*, \\
    \langle\psi|\psi\rangle &\geq 0, \\
    \langle\psi|\psi\rangle &= 0 \Leftrightarrow \ket{\psi} = 0, \\
    \langle\psi|a\varphi_1 + b\varphi_2\rangle &= a\langle\psi|\varphi_1\rangle + b\langle\psi|\varphi_2\rangle,
     \end{align*}
    
    and this inner product induces a norm

\[
\|\cdot\|: V \longrightarrow \mathbb{R}, \]
\[
\ket{\psi} \longmapsto \sqrt{\langle \psi | \psi \rangle}.
\]
   
\end{enumerate}
\end{definition}

\begin{proposition}
Let \((V, ( \cdot, \cdot ))\) be a finite-dimensional complex inner product space of dimension \(n\). Then \((V, ( \cdot, \cdot ))\) is isometrically isomorphic to \((\mathbb{C}^n, ( \cdot, \cdot )_0)\), where \(( \cdot, \cdot )_0\) denotes the standard inner product on \(\mathbb{C}^n\) defined by:
\[
(\ket{\psi},\ket{\varphi})_0:=\langle \psi| \varphi \rangle_0 = \sum_{i=1}^n \psi_i^* \varphi_i \quad \text{for all } \ket{\psi}, \ket{\varphi} \in \mathbb{C}^n.
\]
\end{proposition}
\begin{remark}
    Since we are only concerned with finite-dimensional inner product spaces in this work, the above proposition allows us (without loss of generality) to represent any such space as \(\mathbb{C}^n\) equipped with the standard inner product. Accordingly, we will work in \(\mathbb{C}^n\) throughout the remainder of this project and fix the standard basis as needed.
\end{remark}

\begin{definition}
Let \( A : \mathbb{C}^n \to \mathbb{C}^n \) be a linear map. The Hermitian conjugate of \( A \), denoted \( A^\dagger \), is the unique linear map \( A^\dagger : \mathbb{C}^n \to \mathbb{C}^n \) satisfying
\[
\braket{ A^\dagger \psi | \varphi } = \braket{\psi| A \varphi } \quad \forall \ket{\psi}, \ket{\varphi} \in \mathbb{C}^n.
\]
If \( A^\dagger = A \), then \( A \) is called Hermitian. In the matrix representation, the Hermitian conjugate corresponds to the conjugate transpose of the matrix representing \(A\): 
\[
A^\dagger = (A^*)^T,
\]
where \(A^*\) denotes the complex conjugate of each entry of \(A\) and \(T\) denotes the transpose.

\end{definition}
\begin{definition}
A linear map \( A : \mathbb{C}^n \to \mathbb{C}^n \) is called normal if
\[
AA^\dagger = A^\dagger A.
\]
\end{definition}
\begin{remark}
    In Dirac notation, it is natural to define \( \ket{\psi}^\dagger \equiv \bra{\psi} \). Consequently, the Hermitian conjugate of the vector \( A\ket{\psi} \) satisfies
    \[
    (A\ket{\psi})^\dagger = \bra{\psi} A^\dagger.
    \]
\end{remark}
\begin{proposition}
    A normal map is Hermitian if and only if it has real eigenvalues.
\end{proposition}

\begin{definition}
We denote by \( \mathcal{M}_n(\mathbb{C}) \) the space of all \( n \times n \) matrices with complex entries. More generally, \( \mathcal{M}_{m \times n}(\mathbb{C}) \) refers to the set of all \( m \times n \) complex matrices.
\end{definition}

\begin{definition}
Given a matrix \( A \in \mathcal{M}_n(\mathbb{C}) \), the spectral norm is the matrix norm subordinate to the Euclidean norm \( \| \cdot \|_2 \) (which is induced by the standard inner product on $\mathbb{C}^n$), defined by
\[
\|A\| := \sup_{\|\,\ket{\psi}\| = 1} \|A|\psi\rangle\| = \sup_{\langle \psi | \psi \rangle = 1} \sqrt{ \langle \psi | A^\dagger A | \psi \rangle }.
\]
\end{definition}

\begin{proposition}
Let \( A \in \mathcal{M}_n(\mathbb{C}) \) be a normal matrix. Then, the spectral norm of \( A \) is equal to the magnitude of its largest eigenvalue:
\[
\|A\| = \max_{\lambda \in \sigma(A)} |\lambda|,
\]
where \( \sigma(A) \) denotes the spectrum (set of eigenvalues) of \( A \).
\end{proposition}

\begin{definition}
A linear map \( U : \mathbb{C}^n \to \mathbb{C}^n \) is called unitary if
\[
\langle U\psi | U\varphi \rangle = \langle \psi | \varphi \rangle \quad \forall\, \ket{\psi}, \ket{\varphi} \in \mathbb{C}^n.
\]
\end{definition}

\begin{proposition}
Let \( U : \mathbb{C}^n \to \mathbb{C}^n \) be a linear map. The following statements are equivalent:
\begin{enumerate}
    \item \( U \) is unitary.
    \item $U$ is normal and all eigenvalues of \(U\) have modulus 1.
    \item \( U^\dagger U = U U^\dagger = I \), where \(I\) is the identity map on \( \mathbb{C}^n \).
    \item \( \| U \ket{\psi} \| = \| \ket{\psi} \| \quad \text{for all } \ket{\psi} \in \mathbb{C}^n \).
\end{enumerate}
\end{proposition}

\begin{definition}
Let \( A, B : \mathbb{C}^n \to \mathbb{C}^n \) be linear maps. The commutator of \(A\) and \(B\) is defined as
\[
[A, B] \coloneqq AB - BA.
\]
If \([A, B] = 0\), we say that \(A\) and \(B\) commute.
\end{definition}

\begin{definition}
Let \( \ket{\psi}, \ket{\varphi} \in \mathbb{C}^n \). The outer product \( \ket{\varphi}\bra{\psi} \) is the linear map \( \mathbb{C}^n \to \mathbb{C}^n \) defined by its action on any \( \ket{\psi'} \in \mathbb{C}^n \) as:
\[
\left( \ket{\varphi}\bra{\psi} \right)\ket{\psi'} \coloneqq \braket{\psi | \psi'} \ket{\varphi}.
\]
\end{definition}

\begin{lemma}[Completeness Relation]
Let \( \{ \ket{e_i} \}_{i=1}^n \) be an orthonormal basis for \( \mathbb{C}^n \). Then
\[
\sum_{i=1}^n \ket{e_i}\bra{e_i} = I.
\]
\end{lemma}

\begin{proof}
Let \( \ket{\psi} \in \mathbb{C}^n \). Since \( \{ \ket{e_i} \}_{i=1}^n \) is an orthonormal basis, we can write
\[
\ket{\psi} = \sum_{i=1}^n \psi_i \ket{e_i}, \quad \text{where } \psi_i = \braket{e_i | \psi}.
\]
Consider the action of the map \( \sum_{i=1}^n \ket{e_i}\bra{e_i} \) on \( \ket{\psi} \):
\[
\left( \sum_{i=1}^n \ket{e_i}\bra{e_i} \right)\ket{\psi} = \sum_{i=1}^n \braket{e_i | \psi} \ket{e_i}  = \sum_{i=1}^n \psi_i \ket{e_i} = \ket{\psi}.
\]
Thus, the map acts as the identity on an arbitrary vector \( \ket{\psi} \in \mathbb{C}^n \). Therefore,
\[
\sum_{i=1}^n \ket{e_i}\bra{e_i} = I.
\]
\end{proof}

\begin{lemma}
Let \( \ket{v}, \ket{w} \in \mathbb{C}^n \). Then the Hermitian conjugate of their outer product satisfies:
\[
\left( \ket{w}\bra{v} \right)^\dagger = \ket{v}\bra{w}.
\]
\begin{proof}
We verify the identity by checking that \( \ket{v}\bra{w} \) satisfies the defining property of the Hermitian conjugate of \( \ket{w}\bra{v} \). Let \( \ket{\varphi}, \ket{\psi} \in \mathbb{C}^n \). Then:
\begin{align*}
\braket{(|v\rangle\!\langle w|) \varphi | \psi} &=  \langle\langle w|\varphi\rangle v|\psi\rangle  \\ &=  \langle\langle \varphi|w\rangle^* v|\psi\rangle  \\ &= \langle \varphi|w\rangle \langle v|\psi\rangle \\ &=  \langle v|\psi\rangle \langle \varphi|w\rangle \\ &= \langle \varphi|\langle v|\psi\rangle w\rangle \\
&=
 \braket{ \varphi |(|w\rangle\!\langle v|) \psi} 
.
\end{align*}
\end{proof}

\end{lemma}
\begin{theorem}
Let \( A \in \mathcal{M}_n(\mathbb{C}) \). The following statements are equivalent:
\begin{enumerate}
    \item \( A \) is normal.
    \item There exists a unitary matrix \( U \in \mathcal{M}_n(\mathbb{C}) \) such that:
    \[
    U^{-1}AU = D \quad \text{where \( D \) is diagonal.}
    \]
    \item There exists an orthonormal basis of \( \mathbb{C}^n \) consisting of the eigenvectors of \( A \).
    \end{enumerate}
\end{theorem}
\begin{corollary}
If \( A \in \mathcal{M}_n(\mathbb{C}) \) is normal, then there exists an orthonormal basis \( \{ \ket{e_i} \}_{i=1}^n \) of \( \mathbb{C}^n \) consisting of eigenvectors of \( A \) and corresponding eigenvalues \( \lambda_i \in \mathbb{C} \), such that:
\[
A = \sum_{i=1}^n \lambda_i \ket{e_i}\bra{e_i}.
\]
This is called the spectral decomposition of \( A \).
\end{corollary}

\begin{definition}\label{function}
Let \( f : \mathbb{C} \to \mathbb{C} \) be a function and let \( A \) be a normal map with spectral decomposition
\[
A = \sum_i \lambda_i \ket{e_i}\bra{e_i}.
\]
We define \( f(A) \) by
\[
f(A) \coloneqq \sum_i f(\lambda_i) \ket{e_i}\bra{e_i}.
\]
\end{definition}
\begin{observation}
This definition agrees with the power series expansion when \( f \) is analytic. Suppose \( f(z) = \sum_{n=0}^\infty a_n z^n \) converges on an open neighborhood of the spectrum \( \sigma(A) \). Then:
\begin{align*}
f(A) &= \sum_{n=0}^\infty a_n A^n 
= \sum_{n=0}^\infty a_n \left( \sum_i \lambda_i^n \ket{e_i}\bra{e_i} \right) 
= \sum_i \left( \sum_{n=0}^\infty a_n \lambda_i^n \right) \ket{e_i}\bra{e_i}
= \sum_i f(\lambda_i) \ket{e_i}\bra{e_i},
\end{align*}
which coincides with the spectral definition. Hence, the spectral calculus extends the notion of applying functions to matrices beyond analytic functions to arbitrary functions defined on the spectrum of \( A \).
\end{observation}

\begin{definition}
    Let \( A \in \mathcal{M}_{m \times n}(\mathbb{C)}\) and \( B \in \mathcal{M}_{p \times q}(\mathbb{C)}\). The Kronecker product of \( A \) and \( B \), denoted \( A \otimes B \), is the \( mp \times nq \) block matrix defined as:

\[
A \otimes B = 
\begin{bmatrix}
a_{11}B & a_{12}B & \cdots & a_{1n}B \\
a_{21}B & a_{22}B & \cdots & a_{2n}B \\
\vdots & \vdots & \ddots & \vdots \\
a_{m1}B & a_{m2}B & \cdots & a_{mn}B
\end{bmatrix},
\]
where each \( a_{ij}B \) represents the scalar multiplication of the matrix \( B \) by the entry \( a_{ij} \) of \( A \).
In particular, let \( \ket{\psi} \in \mathbb{C}^m \) and \( \ket{\varphi} \in \mathbb{C}^p \). The Kronecker product \( \ket{\psi} \otimes \ket{\varphi} \) is the column vector of size $mp$ obtained as:
\[
\ket{\psi} \otimes \ket{\varphi} = 
\begin{bmatrix}
\psi_1 \ket{\varphi} \\
\psi_2 \ket{\varphi} \\
\vdots \\
\psi_m \ket{\varphi}
\end{bmatrix},
\]
where each \( \psi_i \ket{\varphi} \) is the scalar multiplication of the vector \( \ket{\varphi} \) by the \( i \)-th component \( \psi_i \) of \( \ket{\psi} \).
\end{definition}

\begin{proposition}
Let \( \ket{\psi}, \ket{\psi'_1}, \ket{\psi'_2} \in \mathbb{C}^n \), \( \ket{\varphi}, \ket{\varphi'_1}, \ket{\varphi'_2} \in \mathbb{C}^m \) and let \( z \in \mathbb{C} \). Then the Kronecker product satisfies:
\begin{enumerate}
    \item Scalar multiplication:
    \[
    z \left( \ket{\psi} \otimes \ket{\varphi} \right) = \left( z\ket{\psi} \right) \otimes \ket{\varphi} = \ket{\psi} \otimes \left( z\ket{\varphi} \right).
    \]
    
    \item Linearity in the first argument:
    \[
    \left( \ket{\psi'_1} + \ket{\psi'_2} \right) \otimes \ket{\varphi} = \ket{\psi'_1} \otimes \ket{\varphi} + \ket{\psi'_2} \otimes \ket{\varphi}.
    \]
    
    \item Linearity in the second argument:
    \[
    \ket{\psi} \otimes \left( \ket{\varphi'_1} + \ket{\varphi'_2} \right) = \ket{\psi} \otimes \ket{\varphi'_1} + \ket{\psi} \otimes \ket{\varphi'_2}.
    \]
\end{enumerate}
\end{proposition}

\begin{proof}
We prove each part in turn:
\begin{enumerate}
    \item Scalar Multiplication:
    \begin{align*}
    z \left( \ket{\psi} \otimes \ket{\varphi} \right) &= z 
    \begin{bmatrix}
    \psi_1 \ket{\varphi} \\
    \vdots \\
    \psi_n \ket{\varphi}
    \end{bmatrix} =
    \begin{bmatrix}
    z \psi_1 \ket{\varphi} \\
    \vdots \\
    z \psi_n \ket{\varphi}
    \end{bmatrix}, \\
    \left( z \ket{\psi} \right) \otimes \ket{\varphi} &= 
    \begin{bmatrix}
    (z \psi_1) \ket{\varphi} \\
    \vdots \\
    (z \psi_n) \ket{\varphi}
    \end{bmatrix} =
    \begin{bmatrix}
    z \psi_1 \ket{\varphi} \\
    \vdots \\
    z \psi_n \ket{\varphi}
    \end{bmatrix}, \\
    \ket{\psi} \otimes \left( z \ket{\varphi} \right) &= 
    \begin{bmatrix}
    \psi_1 (z \ket{\varphi}) \\
    \vdots \\
    \psi_n (z \ket{\varphi})
    \end{bmatrix} =
    \begin{bmatrix}
    z \psi_1 \ket{\varphi} \\
    \vdots \\
    z \psi_n \ket{\varphi}
    \end{bmatrix}.
    \end{align*}
    All three expressions are identical, proving the property.

    \item Linearity in the First Argument:
    Let \(\ket{\psi'_1} = [\psi'_{1,i}]_{i=1}^n\), \(\ket{\psi'_2} = [\psi'_{2,i}]_{i=1}^n\). Then:
    \begin{align*}
    \left( \ket{\psi'_1} + \ket{\psi'_2} \right) \otimes \ket{\varphi} &= 
    \begin{bmatrix}
    (\psi'_{1,1} + \psi'_{2,1}) \ket{\varphi} \\
    \vdots \\
    (\psi'_{1,n} + \psi'_{2,n}) \ket{\varphi}
    \end{bmatrix} \\ &=
    \begin{bmatrix}
    \psi'_{1,1} \ket{\varphi} + \psi'_{2,1} \ket{\varphi} \\
    \vdots \\
    \psi'_{1,n} \ket{\varphi} + \psi'_{2,n} \ket{\varphi}
    \end{bmatrix} \\
    &= 
    \begin{bmatrix}
    \psi'_{1,1} \ket{\varphi} \\
    \vdots \\
    \psi'_{1,n} \ket{\varphi}
    \end{bmatrix} +
    \begin{bmatrix}
    \psi'_{2,1} \ket{\varphi} \\
    \vdots \\
    \psi'_{2,n} \ket{\varphi}
    \end{bmatrix} \\ &= \ket{\psi'_1} \otimes \ket{\varphi} + \ket{\psi'_2} \otimes \ket{\varphi}.
    \end{align*}

    \item Linearity in the Second Argument:
    The proof is analogous.
\end{enumerate}
\end{proof}

\begin{proposition}
Let \( A \in \mathcal{M}_{m \times n}(\mathbb{C}) \), \( B \in \mathcal{M}_{p \times q}(\mathbb{C}) \),
\( C \in \mathcal{M}_{n \times k}(\mathbb{C}) \) and \( D \in \mathcal{M}_{q \times r}(\mathbb{C}) \).
Then the following identity holds:
\[
(A \otimes B)(C \otimes D) = (AC) \otimes (BD).
\]
\end{proposition}

\begin{proof}
We compute the matrix product explicitly using the Kronecker product structure:
\begin{align*}
(A \otimes B)(C \otimes D) &= 
\begin{bmatrix}
a_{11}B & \cdots & a_{1n}B \\
\vdots & \ddots & \vdots \\
a_{m1}B & \cdots & a_{mn}B
\end{bmatrix}
\begin{bmatrix}
c_{11}D & \cdots & c_{1k}D \\
\vdots & \ddots & \vdots \\
c_{n1}D & \cdots & c_{nk}D
\end{bmatrix} \\
&= 
\begin{bmatrix}
\sum_{l=1}^n a_{1l}B c_{l1}D & \cdots & \sum_{l=1}^n a_{1l}B c_{lk}D \\
\vdots & \ddots & \vdots \\
\sum_{l=1}^n a_{ml}B c_{l1}D & \cdots & \sum_{l=1}^n a_{ml}B c_{lk}D
\end{bmatrix} \\
&= 
\begin{bmatrix}
\left(\sum_{l=1}^n a_{1l}c_{l1}\right)BD & \cdots & \left(\sum_{l=1}^n a_{1l}c_{lk}\right)BD \\
\vdots & \ddots & \vdots \\
\left(\sum_{l=1}^n a_{ml}c_{l1}\right)BD & \cdots & \left(\sum_{l=1}^n a_{ml}c_{lk}\right)BD
\end{bmatrix} \\
&= 
\begin{bmatrix}
(AC)_{11}BD & \cdots & (AC)_{1k}BD \\
\vdots & \ddots & \vdots \\
(AC)_{m1}BD & \cdots & (AC)_{mk}BD
\end{bmatrix} \\
&= (AC) \otimes (BD). 
\end{align*}
\end{proof}

\begin{corollary}
Let \( A \in \mathcal{M}_{m \times n}(\mathbb{C)}\), \( B \in \mathcal{M}_{p \times q}(\mathbb{C)}\), $|\psi\rangle \in \mathbb{C}^n$ and $|\varphi\rangle \in \mathbb{C}^q$. Then the following identity holds:
\[
(A \otimes B)(|\psi\rangle \otimes |\varphi\rangle) = A|\psi\rangle \otimes B|\varphi\rangle.
\]
\end{corollary}

\begin{definition}
Let \( V \cong \mathbb{C}^n \) and \( W \cong \mathbb{C}^m \) be finite-dimensional complex vector spaces. The tensor product space \( V \otimes W \) is the complex vector space spanned by all formal linear combinations of elementary tensors of the form \( \ket{\psi} \otimes \ket{\varphi} \), where \( \ket{\psi} \in V \) and \( \ket{\varphi} \in W \):

\[
V \otimes W := \text{Span}_{\mathbb{C}} \left\{ \ket{\psi} \otimes \ket{\varphi} \mid \ket{\psi} \in V, \ket{\varphi} \in W \right\}.
\]
\end{definition}

\begin{remark}
    In this context, we define the tensor product space \( V \otimes W \) as the span of Kronecker products of vectors from \( V \) and \( W \). This is a natural and practical construction for finite-dimensional spaces. While the general tensor product is defined abstractly via an universal property and applies to arbitrary vector spaces (including infinite-dimensional ones), the Kronecker product serves as an explicit realization in the case of finite-dimensional vector spaces over $\mathbb{C}$, matching both the vector space structure and the required properties of the tensor product.
\end{remark}

\begin{observation}
In the tensor product space \( V \otimes W \), where \( V \cong \mathbb{C}^n \) and \( W \cong \mathbb{C}^m \) are finite-dimensional complex vector spaces, we can define a natural inner product on \( V \otimes W \). Given vectors \( \ket{\psi} \otimes \ket{\varphi}, \ket{\psi'} \otimes \ket{\varphi'} \in V \otimes W \), the inner product is defined as:

\[
\braket{\psi \otimes \varphi | \psi' \otimes \varphi'}_{V \otimes W} = \braket{\psi | \psi'}_V \braket{\varphi | \varphi'}_W,
\]
where \( \braket{\psi | \psi'}_V \) is the inner product in \( V \) and \( \braket{\varphi | \varphi'}_W \) is the inner product in \( W \). 
\end{observation}

\begin{proposition}
Let \( V \) and \( W \) be finite-dimensional inner product spaces over \( \mathbb{C} \), with \(\dim(V) = n\) and \(\dim(W) = m\). Then the tensor product space \( V \otimes W \) has dimension:
\[
\dim(V \otimes W) = \dim(V) \times \dim(W) = nm.
\]
\end{proposition}

\begin{proof}
Let \(\{\ket{\psi_1}, \dots, \ket{\psi_n}\}\) be an orthonormal basis for \( V \) and \(\{\ket{\varphi_1}, \dots, \ket{\varphi_m}\}\) be an orthonormal basis for \( W \). We prove that:
\[
\{\ket{\psi_i} \otimes \ket{\varphi_j} | 1 \leq i \leq n, \, 1 \leq j \leq m\},
\]
is a basis for \( V \otimes W \).

\begin{enumerate}

\item Spanning Property:
Any vector in \( V \otimes W \) is a linear combination of simple tensors \(\ket{\psi} \otimes \ket{\varphi}\) where \(\ket{\psi} \in V\) and \(\ket{\varphi} \in W\). Express \(\ket{\psi}\) and \(\ket{\varphi}\) in terms of their bases:
\[
\ket{\psi} = \sum_{i=1}^n a_i \ket{\psi_i}, \quad \ket{\varphi} = \sum_{j=1}^m b_j \ket{\varphi_j}.
\]
By the bilinearity of the tensor product, we have:
\[
\ket{\psi} \otimes \ket{\varphi} = \left(\sum_{i=1}^n a_i \ket{\psi_i}\right) \otimes \left(\sum_{j=1}^m b_j \ket{\varphi_j}\right) = \sum_{i=1}^n \sum_{j=1}^m a_i b_j (\ket{\psi_i}\otimes \ket{\varphi_j}).
\]
Thus, \(\{\ket{\psi_i} \otimes \ket{\varphi_j}\}\) spans \( V \otimes W \).

\item Linear Independence:
Suppose there exists a linear combination:
\[
\sum_{i=1}^n \sum_{j=1}^m c_{ij} (\ket{\psi_i} \otimes \ket{\varphi_j}) = 0.
\]
For any fixed $(k,l)$ where $1 \leq k \leq n$ and $1 \leq l \leq m$, take the inner product of both sides with $\ket{\psi_k} \otimes \ket{\varphi_l}$:

\[0 =
\left\langle \ket{\psi_k} \otimes \ket{\varphi_l} , \sum_{i,j} c_{ij} (\ket{\psi_i} \otimes \ket{\varphi_j}) \right\rangle = \sum_{i,j} c_{ij} \langle  \psi_k | \psi_i \rangle \langle  \varphi_l | \varphi_j \rangle.
\]

Since the bases are orthonormal, $ \langle  \psi_k | \psi_i \rangle = \delta_{ki}$ and $\langle  \varphi_l | \varphi_j \rangle = \delta_{lj}$, so:

\[
0 = \sum_{i,j} c_{ij} \delta_{ki} \delta_{lj} = c_{kl}.
\]

This holds for all $1 \leq k \leq n$ and $1 \leq l \leq m$, proving all coefficients $c_{ij} = 0$. Therefore, the set $\{\ket{\psi_i} \otimes \ket{\varphi_j}\}$ is linearly independent.
\\\\
Since the set \(\{\ket{\psi_i} \otimes \ket{\varphi_j}\}\) is a basis for \( V \otimes W \) and its size is \( n \times m \),
\[
\dim(V \otimes W) = \dim(V) \times \dim(W) = nm.
\]
 
\end{enumerate}
\end{proof}
\begin{observation}
        Since $V \cong C^n$ and $W \cong C^m$, we have $V \otimes W \cong C^{nm}$.
\end{observation}

\begin{theorem}[Singular value decomposition]
Let \( A \in \mathcal{M}_n(\mathbb{C)} \). Then there exist unitary matrices \( U \), \( V \) and a diagonal matrix \( D \) with non-negative entries such that  
\[ A = U D V. \]  
The diagonal elements of \( D \) are called the singular values of \( A \).
\end{theorem}
\begin{theorem}[Schmidt decomposition] \label{schmidt}
Let \( |\Psi\rangle \in \mathbb{C}^n \otimes \mathbb{C}^n\). Then there exist orthonormal sets \( \{|\psi_i\rangle\} \in\mathbb{C}^n  \) and \(\{|\varphi_i\rangle\} \in\mathbb{C}^n  \) such that
\[
|\Psi\rangle = \sum_i \lambda_i |\psi_i\rangle \otimes|\varphi_i\rangle.
\]
\end{theorem}
\begin{proof}
Let \(\{|e_j\rangle\}\) and \(\{|\tilde{e}_k\rangle\}\) be fixed orthonormal bases for \(\mathbb{C}^n\). The state \(|\Psi\rangle \in \mathbb{C}^n \otimes \mathbb{C}^n\) can be expanded as:
\[
|\Psi\rangle = \sum_{j,k} a_{jk} |e_j\rangle \otimes |\tilde{e}_k\rangle,
\]
where \(a_{jk}\) are complex numbers forming a matrix \(A\). By the singular value decomposition theorem, \(A = U D V\), where \(D\) is a diagonal matrix with non-negative entries \(\lambda_i\) and \(U\) and \(V\) are unitary matrices. Substituting this into the expansion:
\[
|\Psi\rangle = \sum_{j,k} \left( \sum_i u_{ji} \lambda_i v_{ik} \right) |e_j\rangle \otimes |\tilde{e}_k\rangle = \sum_i \lambda_i \left( \sum_j u_{ji} |e_j\rangle \right) \otimes \left( \sum_k v_{ik} |\tilde{e}_k\rangle \right).
\]
Define the states:
\[
|\psi_i\rangle \equiv \sum_j u_{ji} |e_j\rangle, \quad |\varphi_i\rangle \equiv \sum_k v_{ik} |\tilde{e}_k\rangle.
\]
Thus:
\[
|\Psi\rangle = \sum_i \lambda_i |\psi_i\rangle \otimes |\varphi_i\rangle.
\]
To verify orthonormality of $\{|\psi_i\rangle\}$, compute:
\[
\langle \psi_i | \psi_{i'} \rangle = \sum_{j,j'} u_{ji}^* u_{j'i'} \langle e_j | e_{j'} \rangle = \sum_j u_{ji}^* u_{ji'} = (U^\dagger U)_{i'i} = \delta_{i'i},
\]
since $U^\dagger U = I$. Orthonormality of $\{|\varphi_i\rangle\}$ is shown analogously.
\end{proof}

\begin{proposition}
Let \( A, B \in \mathcal{M}_{n}(\mathbb{C}) \). Then:
\[
\|A \otimes B\| = \|A\| \cdot \|B\|.
\]
\end{proposition}

\begin{proof}
Let $|\Psi\rangle \in \mathbb{C}^n \otimes \mathbb{C}^n$. Using Theorem \ref{schmidt} we can, without loss of generality, write  \[
|\Psi\rangle = \sum_i \lambda_i |\psi_i\rangle \otimes|\varphi_i\rangle,
\] with  the sets \( \{|\psi_i\rangle\} \) and \(\{|\varphi_i\rangle\} \) orthonormal. It follows that the set \(\{ |\psi_i\rangle \otimes |\varphi_i\rangle \}\) is also orthonormal. Then, by the generalized Pythagorean theorem:
\[
\||\Psi\rangle\|^2 = \left\| \sum_{i} \lambda_i|\psi_i\rangle \otimes |\varphi_i\rangle \right\|^2 = \sum_{i} \lambda_i^2\| |\psi_i\rangle \otimes |\varphi_i\rangle \|^2 = \sum_{i} \lambda_i^2 \| |\psi_i\rangle \|^2.
\]
Similarly,
\[
(A \otimes I)|\Psi\rangle = \sum_{i} \lambda_iA |\psi_i\rangle \otimes |\varphi_i\rangle,
\]
and the set \(\{ A|\psi_j\rangle \otimes |\varphi_j\rangle \}\) remains orthogonal, so
\[
\|(A \otimes I)|\Psi\rangle\|^2 = \sum_{i} \lambda_i^2\| A|\psi_i\rangle \|^2 \leq \sum_{i} \lambda_i^2\|A\|^2 \| |\psi_i\rangle \|^2 = \|A\|^2 \||\Psi\rangle\|^2.
\]
Thus, by the definition of the spectral norm,
\[
\|A \otimes I\| \leq \|A\|.
\]
By the same argument,
\[
\|I \otimes B\| \leq \|B\|.
\]
Since
\[
A \otimes B = (A \otimes I)(I \otimes B),
\]
it follows that
\[
\|A \otimes B\| \leq \|A \otimes I\| \cdot \|I \otimes B\| \leq \|A\| \cdot \|B\|.
\]
For the reverse inequality, let \(\epsilon > 0\). By definition of the spectral norm, there exist unit vectors \(|\psi\rangle,|\varphi\rangle \in \mathbb{C}^n\) such that
\[
\|A|\psi\rangle\| > \|A\| - \epsilon, \quad \|B|\varphi\rangle\| > \|B\| - \epsilon.
\]
Then \(|\psi\rangle \otimes |\varphi\rangle\) is a unit vector in \(\mathbb{C}^n \otimes \mathbb{C}^n\) and
\[
\|(A \otimes B)(|\psi\rangle \otimes |\varphi\rangle)\| = \|A|\psi\rangle \otimes B|\varphi\rangle\| = \|A|\psi\rangle\| \cdot \|B|\varphi\rangle\| > (\|A\| - \epsilon)(\|B\| - \epsilon).
\]
Taking the limit \(\epsilon \to 0^+\) yields
\[
\|A \otimes B\| \geq \|A\| \cdot \|B\|.
\]
Combining both inequalities completes the proof.
\end{proof}
\begin{corollary}
Let \( A_1, ..., A_n\in \mathcal{M}_m(\mathbb{C})\). Then:
\[
\left\| A_1 \otimes A_2 \otimes \cdots \otimes A_n \right\| = \prod_{i=1}^n \|A_i\|.
\]
\end{corollary}

\begin{proposition}
Let \( A \in \mathcal{M}_{n}(\mathbb{C)}\) and \( B \in \mathcal{M}_{m}(\mathbb{C)}\). Then:
\begin{enumerate}
    \item  If $A$ and $B$ are Hermitian, then $A \otimes B$ is Hermitian.
    \item  If $A$ and $B$ are unitary, then $A \otimes B$ is unitary.
\end{enumerate}
\end{proposition}

\begin{proof}
\begin{enumerate}
\item First, we show that if \( A \) and \( B \) are Hermitian, then \( A \otimes B \) is also Hermitian. The conjugate and transpose of \( A \otimes B \) are computed as follows:
    \begin{align*}
(A\otimes B)^{*} &= 
\begin{bmatrix}
a_{11}B & a_{12}B & \dots & a_{1n}B \\ 
a_{21}B & a_{22}B & \dots & a_{2n}B \\ 
\vdots & \vdots & \vdots & \vdots \\ 
a_{m1}B & a_{m2}B & \dots & a_{mn}B
\end{bmatrix}^{*} \\
&= 
\begin{bmatrix}
a^{*}_{11}B^{*} & a^{*}_{12}B^{*} & \dots & a^{*}_{1n}B^{*} \\ 
a^{*}_{21}B^{*} & a^{*}_{22}B^{*} & \dots & a^{*}_{2n}B^{*} \\ 
\vdots & \vdots & \vdots & \vdots \\ 
a^{*}_{m1}B^{*} & a^{*}_{m2}B^{*} & \dots & a^{*}_{mn}B^{*}
\end{bmatrix} \\
&= A^{*}\otimes B^{*}.
\end{align*}

\begin{align*}
(A\otimes B)^{T} &= 
\begin{bmatrix}
a_{11}B & a_{12}B & \dots & a_{1n}B \\ 
a_{21}B & a_{22}B & \dots & a_{2n}B \\ 
\vdots & \vdots & \vdots & \vdots \\ 
a_{m1}B & a_{m2}B & \dots & a_{mn}B
\end{bmatrix}^{T} \\
&= 
\begin{bmatrix}
a_{11}B^{T} & a_{21}B^{T} & \dots & a_{n1}B^{T} \\ 
a_{12}B^{T} & a_{22}B^{T} & \dots & a_{n2}B^{T} \\ 
\vdots & \vdots & \vdots & \vdots \\ 
a_{1m}B^{T} & a_{2m}B^{T} & \dots & a_{nm}B^{T}
\end{bmatrix} \\
&= A^{T}\otimes B^{T}.
\end{align*}
The first statement follows from the former two relations. If $A$ and $B$ are Hermitian then:
\begin{align*}
(A\otimes B)^{\dagger} &= \left((A\otimes B)^{T}\right)^{*} \\
&= (A^{T}\otimes B^{T})^{*} \\
&= (A^{T})^{*}\otimes(B^{T})^{*} \\
&= A^{\dagger}\otimes B^{\dagger} \\
&= A \otimes B.
\end{align*}
\item Next, we show that if \( A \) and \( B \) are unitary, then \( A \otimes B\) is unitary
\begin{align*}
(A \otimes B)^\dagger (A \otimes B) &= (A^\dagger \otimes B^\dagger)(A \otimes B) \\
&= (A^\dagger A \otimes B^\dagger B) \\
&= I \otimes I.
\end{align*}
\end{enumerate}
\end{proof}

\begin{observation}
    The notation $\ket{\psi}^{\otimes k}$ will be used to represent $\ket{\psi}$ tensored with itself $k$ times. For example $\ket{\psi}^{\otimes 2} = \ket{\psi} \otimes \ket{\psi}$. An analogous notation is also used for linear maps and vector spaces.
\end{observation}

\begin{definition}
The following matrices are referred to as the Pauli matrices:
\[
 X \equiv \begin{bmatrix} 0 & 1 \\ 1 & 0 \end{bmatrix}, \quad
 Y \equiv \begin{bmatrix} 0 & -i \\ i & 0 \end{bmatrix}, \quad
 Z \equiv \begin{bmatrix} 1 & 0 \\ 0 & -1 \end{bmatrix}.
\]
\end{definition}
These matrices are Hermitian, unitary and they play a fundamental role in quantum computation and quantum information.
\pagebreak

\chapter{Quantum Computation}
\section{The Postulates of Quantum Mechanics}
Quantum mechanics arose as a mathematical framework to describe and predict the outcome of various experiments carried out between the late 19th century and the early 20th century, which could not be explained by classical physics. Examples of such experiments include the photoelectric effect and the Stern-Gerlach experiment. The postulates of quantum mechanics did not arise from intuitive reasoning, as is often the case in physics, but were instead deduced retrospectively from experimental observations that resisted classical interpretation. It is therefore not always trivial to see how they provide a connection between the physical world and mathematical formalism. In this section, we will provide a brief but intuitive introduction to these postulates, aiming to convey why they are formulated in the way they are, without delving into the full details, as this is not the central focus of this work.
 To help us, we will see the role they play in describing the spin of electrons. 
\\\\
While we will not delve deeply into the full physical meaning of electron spin, for the purposes of this work, it can be understood operationally: spin refers to the direction in which a particle is deflected when passed through a magnetic field. Specifically, a magnetic field oriented along a given axis (e.g., the \( z \)-axis) enables measurement of the spin component along that direction. In the Stern-Gerlach experiment (see Figure~\ref{fig:sterngerlach}), a beam of electrons was passed through a spatially varying (inhomogeneous) magnetic field aligned along a particular axis. \footnote{In the original experiment, neutral silver atoms (not electrons) were used. However, the observed quantization of angular momentum was later interpreted as evidence of the intrinsic spin of the unpaired electron in the silver atom. For pedagogical clarity, we refer to electrons directly in this discussion.} This inhomogeneous field exerts a force on the magnetic moment of the electrons, causing them to deflect. Contrary to classical expectations of a continuous distribution of deflections, the beam split into exactly two distinct paths. The electrons were deflected either upward or downward, corresponding to two discrete spin states: ``spin-up'' and ``spin-down.'' This demonstrated that spin is quantized and, when measured along any axis, can take only one of two values.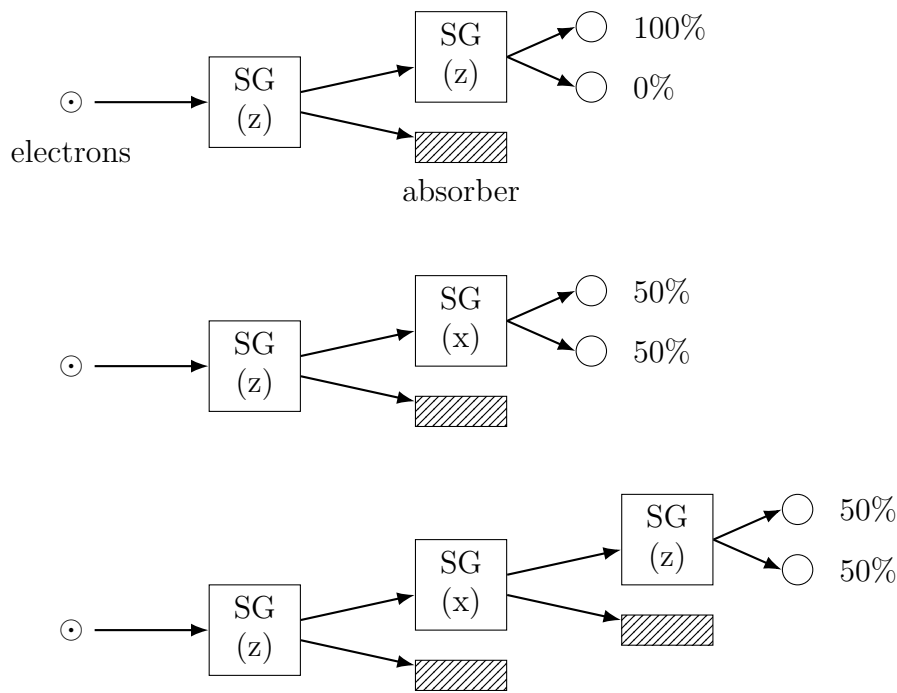
\begin{figure}[H]
    \centering
\begin{tikzpicture}[node distance=1cm and 1.5cm]

\begin{scope}[yshift=-3.5cm]
\node (source1) at (0,0) {$\odot$};
\node[sg, right=of source1] (sgz1) {SG\\(z)};
\node[absorber, right=of sgz1, yshift=-0.6cm] (abs1) {};
\node[sg, right=of sgz1, yshift=0.6cm] (sgx1) {SG\\(x)};
\node[measurement, right=0.9cm of sgx1, yshift=0.4cm] (m1a) {};
\node[measurement, right=0.9cm of sgx1, yshift=-0.4cm] (m1b) {};

\draw[myarrow] (source1) -- (sgz1);
\draw[myarrow] (sgz1) -- (sgx1) node[midway, above] {};
\draw[myarrow] (sgz1) -- (abs1);
\draw[myarrow] (sgx1.east) -- (m1a.west);
\draw[myarrow] (sgx1.east) -- (m1b.west);
\node[right=0.2cm of m1a] {50\%};
\node[right=0.2cm of m1b] {50\%};
\end{scope}

\begin{scope}[yshift=-7cm]
\node (source2) at (0,0) {$\odot$};
\node[sg, right=of source2] (sgz2) {SG\\(z)};
\node[absorber, right=of sgz2, yshift=-0.6 cm] (abs2) {};
\node[sg, right=of sgz2, yshift=0.6 cm] (sgx2) {SG\\(x)};
\node[absorber, right=of sgx2, yshift=-0.6cm] (abs3) {};
\node[sg, right=of sgx2, yshift=0.6cm] (sgz3) {SG\\(z)};
\node[measurement, right=0.9cm of sgz3, yshift=0.4cm] (m2a) {};
\node[measurement, right=0.9cm of sgz3, yshift=-0.4cm] (m2b) {};

\draw[myarrow] (source2) -- (sgz2);
\draw[myarrow] (sgz2) -- (sgx2) node[midway, above] {};
\draw[myarrow] (sgz2) -- (abs2);
\draw[myarrow] (sgx2) -- (abs3);
\draw[myarrow] (sgx2) -- (sgz3);
\draw[myarrow] (sgz3.east) -- (m2a.west);
\draw[myarrow] (sgz3.east) -- (m2b.west);
\node[right=0.2cm of m2a] {50\%};
\node[right=0.2cm of m2b] {50\%};
\end{scope}


\node (source3) at (0,0) {$\odot$};
\node[below=2pt of source3] {electrons};
\node[sg, right=of source3] (sgz4) {SG\\(z)};
\node[absorber, right= of sgz4, yshift=-0.6 cm] (abs4) {};
\node[below=2pt of abs4] {absorber};
\node[sg, right=of sgz4, yshift = 0.6 cm] (sgz5) {SG\\(z)};
\node[measurement, right=0.9cm of sgz5, yshift=0.4cm] (m3a) {};
\node[measurement, right=0.9cm of sgz5, yshift=-0.4cm] (m3b) {};

\draw[myarrow] (source3) -- (sgz4);
\draw[myarrow] (sgz4) -- (sgz5) node[midway, above] {};
\draw[myarrow] (sgz4) -- (abs4);
\draw[myarrow] (sgz5.east) -- (m3a.west);
\draw[myarrow] (sgz5.east) -- (m3b.west);
\node[right=0.2cm of m3a] {100\%};
\node[right=0.2cm of m3b] {0\%};

\end{tikzpicture}
\caption{Visual representation of the results of the Stern-Gerlach experiment}
\label{fig:sterngerlach}
\end{figure}
 Moreover, the results of spin measurements exhibit inherently probabilistic behavior and depend on the sequence in which measurements are made. For example, if a beam is first filtered to contain only spin-up electrons along the \( z \)-axis (by blocking the downward-deflected component) and then passed through another \( z \)-axis Stern-Gerlach apparatus, all particles are again deflected upward, as expected. However, if instead the filtered spin-up beam is measured along the \( x \)-axis, the result is a 50/50 distribution between spin-up and spin-down along that axis. Remarkably, if we then measure spin along the \( z \)-axis again, the distribution is no longer 100\% spin-up: half the electrons are now found to have spin-down! This sequence of measurements reveals two key features: measurement changes the quantum state, and the outcome is probabilistic, depending both on the state and the measurement ``basis''.
 This idea of different discrete measurement results (up or down), which lead to states determined by the measurement outcome, resembles the close connection between eigenvalues and eigenvectors. Therefore, linear algebra seems like a good mathematical framework for quantum mechanics.

\begin{postulate}[State Space]
Associated with any isolated quantum physical system there is a complex Hilbert space $\mathbb{H}$ known as the state space of the system. The system is completely described by its state vector, which is a unit vector in the system's state space.
\end{postulate}

   \begin{remark}
    A Hilbert space is an inner product space that is complete with respect to the norm induced by the inner product. In other words, every Cauchy-convergent sequence \( (\varphi_{n})_{n \in \mathbb{N}} \subset \mathbb{H} \) has a limit \( \lim_{n \to \infty} \varphi_{n} = \varphi \in \mathbb{H} \) that also lies in \( \mathbb{H} \). In the finite-dimensional inner product spaces that arise in quantum computation, a Hilbert space is essentially the same as an inner product space.
\end{remark}
This way, the spin-up and spin-down states along the z-axis of an electron can be defined as:
\[
\ket{\uparrow_z} = \begin{bmatrix} 1 \\ 0 \end{bmatrix}, \quad \text{and} \quad
\ket{\downarrow_z} = \begin{bmatrix} 0 \\ 1 \end{bmatrix}.
\]
If we associate spin-up with a measurement result of \( +1 \) and spin-down with \( -1 \), we can neatly encode this information about the measurement possibilities using a matrix. In this case, the matrix with eigenvalues and associated eigenvectors as above is the Pauli \( Z \) matrix:
\[
 Z \equiv \begin{bmatrix} 1 & 0 \\ 0 & -1 \end{bmatrix} = 1 \cdot \ket{\uparrow_z}\bra{\uparrow_z} + (-1) \cdot \ket{\downarrow_z}\bra{\downarrow_z}.
\]
We say that the Pauli $Z$ matrix is the observable corresponding to the measurement of the spin component along the $z$-axis.

\begin{postulate}[Measurement]
An observable, that is, a physically measurable quantity of a quantum system is represented by a Hermitian
map $A$ on the system's state space. The possible measurement values of an observable are given by the
spectrum $\sigma(A)$. The probability \( P_{\psi}(\lambda) \) that for a quantum system in the state \( |\psi\rangle \in \mathbb{H} \) a measurement of the observable yields the eigenvalue \(\lambda\) of \(A\) is given with the help of the projection \(P_\lambda\) onto the eigenspace \(\text{Eig}(A,\lambda)\) as
\[
P_{\psi}(\lambda) = \|P_\lambda|\psi\rangle\|^2,
\] where we adopt the convention that \( P_\lambda = 0 \) \( \left(\text{i.e., } \text{Eig}(A, \lambda) = \{0\}  \right) \) when \( \lambda \notin \sigma(A) \). The state of the system after the measurement will be $\frac{P_\lambda|\psi\rangle}{\|P_\lambda|\psi\rangle\|^2}$.
\end{postulate}

To represent the spin states along the \( x \)-axis, and motivated by the outcomes of the Stern-Gerlach experiment, we seek two vectors that satisfy two properties. First, they must be orthonormal, since they must be unit vectors and the probability of measuring spin-down along the \( x \)-axis when the system is in the spin-up \( x \)-axis state should be zero and vice versa. Second, each state must yield a 50/50 probability of measuring spin-up or spin-down along the \( z \)-axis, reflecting the experimental observation that a spin-up \( x \)-axis state yields a symmetric distribution when measured along \( z \). 
These requirements do not uniquely determine the basis, but a conventional and convenient choice is:
\[
 \ket{\uparrow_x} = \frac{1}{\sqrt{2}}(\ket{\uparrow_z} + \ket{\downarrow_z}) = \begin{bmatrix} \frac{1}{\sqrt{2}} \\ \frac{1}{\sqrt{2}} \end{bmatrix}, \quad
\ket{\downarrow_x} = \frac{1}{\sqrt{2}}(\ket{\uparrow_z} - \ket{\downarrow_z}) = \begin{bmatrix} \frac{1}{\sqrt{2}} \\ -\frac{1}{\sqrt{2}} \end{bmatrix},
\]
with corresponding measurement results \( +1 \) and \( -1 \) and the associated observable Pauli \( X \):
\[
 X \equiv \begin{bmatrix} 0 & 1 \\ 1 & 0 \end{bmatrix} = 1 \cdot \ket{\uparrow_x}\bra{\uparrow_x} + (-1) \cdot \ket{\downarrow_x}\bra{\downarrow_x}.
\]

Notice that the projector associated with measuring spin-up along the z-axis is \( P = \ket{\uparrow_z} \bra{\uparrow_z} \) and the probability for measuring spin-up along the z-axis from an electron initially in the spin-up state along the x-axis is:
\[
\| (\ket{\uparrow_z} \bra{\uparrow_z}) \ket{\uparrow_x } \|^2=\| \langle \uparrow_z | \uparrow_x \rangle \ket{\uparrow_z} \|^2= \|\langle \uparrow_z | \uparrow_x \rangle \|^2 = \left\| \frac{1}{\sqrt{2}} \right\|^2 = \frac{1}{2}.
\]
Similarly, for the spin-down along the x-axis:
\[
\| \langle \downarrow_z | \uparrow_x \rangle \|^2 = \frac{1}{2}.
\]

\begin{observation}
        In Postulate 1, the requirement that state vectors are unit vectors ensures that the total probability of all possible outcomes of a measurement sums to 1. This is because the squared magnitudes of the components of the state vector in a certain basis correspond to the probabilities of the measurement outcomes in that basis. Similarly, in Postulate 2, the requirement that the observable is Hermitian guarantees that all possible measurement outcomes are real numbers. This is important, as physical quantities, such as energy levels, must have real values.
\end{observation}

How does the state of a quantum mechanical system change with time? We owe the answer to this question to Erwin Schr\"odinger.

\begin{postulate}[Evolution]
The time evolution of a closed quantum system is described by the Schr\"odinger equation:
\[
i\hbar\frac{d\ket{\psi(t)}}{dt} = H(t)\ket{\psi(t)},
\]
where $\hbar$ is Planck's constant and $H$ is a Hermitian map called the system's Hamiltonian. Often we set $\hbar = 1$ by absorbing it into $H$.
\end{postulate}

\begin{remark}
    In this work, we will only be concerned with the case where the Hamiltonian $H$ is time-independent.
\end{remark}

However, notice that since quantum states are represented by unit vectors, any transformation must be unitary. That is, the norm must be preserved to ensure that the resulting quantum state remains valid.

\begin{postulateprime}
The evolution of a closed quantum system is described by a unitary transformation. That is, the state $\ket{\psi}$ of the system at time $t$ is related to the state $\ket{\psi'}$ at time $t_0$ by a unitary map $U$ that depends only on $t$ and $t_0$:

\[
\ket{\psi'} = U(t,t_0)\ket{\psi}.
\]
\end{postulateprime}
What is the connection between the Hamiltonian picture of dynamics (Postulate 3) and the unitary map picture (Postulate 3')? The answer is provided by writing down the solution to the Schr\"odinger equation, which in the time-independent case is easily verified to be:

\[|\psi(t)\rangle = e^{-iH(t - t_0)} |\psi(t_0)\rangle = U(t, t_0)|\psi(t_0)\rangle,\]
where we define:
\[U(t, t_0) \equiv e^{-iH(t - t_0)}.\]
\begin{remark}
    Without loss of generality, we will take \( t_0 = 0 \), so that the time-evolution map simplifies to \( U(t) = e^{-iHt} \).
\end{remark}

\begin{proposition}
    Let $H$ be a Hermitian matrix, then $e^{-iHt}$ is unitary.
\end{proposition}
\begin{proof}
    Since $H$ is Hermitian, it is normal and thus has a spectral decomposition:
\[
H = \sum_{j} \lambda_j \ket{e_j}\!\bra{e_j},
\]
where all $\lambda_j$ are real and $\{\ket{e_j}\}$ forms an orthonormal basis. 
Substituting the spectral decomposition of $H$ in the expression $e^{-iHt}$ and applying definition \ref{function}:
\[
e^{-iHt} = e^{-i \sum_j \lambda_j \ket{e_j}\!\bra{e_j} t}
= \sum_j e^{-i \lambda_j t }\ket{e_j}\!\bra{e_j}.
\]
The Hermitian conjugate is calculated as:
\[
(e^{-iHt})^\dagger = \left(\sum_j e^{-i \lambda_j t }\ket{e_j}\!\bra{e_j}\right)^\dagger
= \sum_j e^{i \lambda_j t }\ket{e_j}\!\bra{e_j}.
\]

To verify unitarity, we compute the product:
\begin{align*}
e^{-iHt}(e^{-iHt})^\dagger &= \sum_j e^{-i \lambda_j t }\ket{e_j}\!\bra{e_j}
                                  \sum_{j'} e^{i \lambda_{j'} t} \ket{e_{j'}}\!\bra{e_{j'}} \\
&= \sum_j \sum_{j'} e^{-i (\lambda_j - \lambda_{j'}) t}\ket{e_j}\braket{e_j|e_{j'}}\bra{e_{j'}} \\
&= \sum_j e^0 \ket{e_j}\!\bra{e_j} \\
&= \sum_j \ket{e_j}\!\bra{e_j} = I.
\end{align*}

The same calculation shows $(e^{-iHt})^\dagger e^{-iHt} = I$. Therefore, $e^{-iHt}$ is unitary.
\end{proof}

\begin{proposition}
        Every unitary matrix \( U \) can be written as \( U = e^{iK} \) for some Hermitian matrix \( K \).

\end{proposition}
\begin{proof}
    Since \( U \) is unitary, it is normal and thus has a spectral decomposition:
\[
U = \sum_j \lambda_j \ket{e_j}\!\bra{e_j},
\]
    where \( \{\ket{e_j}\} \) is an orthonormal basis and each \( \lambda_j \in \mathbb{C} \) satisfies \( |\lambda_j| = 1 \). Therefore, each eigenvalue can be written as \( \lambda_j = e^{i\theta_j} \) for some real \( \theta_j \in \mathbb{R} \). Hence,
 \[
    U = \sum_j e^{i\theta_j} \ket{e_j}\!\bra{e_j}.
    \]
Define the matrix \( K \) as:
\[
K  = \sum_j \theta_j \ket{e_j}\!\bra{e_j}.
\]
    Since all \( \theta_j \) are real and \( \ket{e_j}\!\bra{e_j} \) are Hermitian, \( K \) is Hermitian:
 \[
    K^\dagger = \sum_j \theta_j \left( \ket{e_j}\!\bra{e_j} \right)^\dagger = \sum_j \theta_j \ket{e_j}\!\bra{e_j} = K.
    \]
 Then,
    \[
    e^{iK} = \sum_j e^{i\theta_j} \ket{e_j}\!\bra{e_j} = U.
    \]
This completes the proof.
\end{proof}
Postulates 3 and 3' provide two equivalent perspectives on the time evolution of quantum states. The Schr\"odinger equation describes the continuous evolution of the system over time, while the unitary transformation framework can be interpreted as representing this continuous process through discrete unitary transformations. This distinction is particularly significant, as it closely mirrors the foundational concept behind digital quantum computing: the approximation of continuous time evolution by discretizing it into a sequence of unitary operations, or time steps.
\\\\
The last postulate explains how to describe quantum systems composed of multiple individual subsystems (e.g., the spins of $n$ different electrons).

\begin{postulate}[Composite Systems]
Let \(\mathbb{H}^A\) and \(\mathbb{H}^B\) be the Hilbert spaces associated with two quantum systems \(A\) and \(B\), respectively. The Hilbert space of the composite system consisting of \(A\) and \(B\) is the tensor product \(\mathbb{H}^A \otimes \mathbb{H}^B\). More generally, for a composite system of \(n\) subsystems with associated Hilbert spaces \(\mathbb{H}^{(1)}, \ldots, \mathbb{H}^{(n)}\), the state space of the total system is given by
\[
\mathbb{H}^{(1)} \otimes \mathbb{H}^{(2)} \otimes \cdots \otimes \mathbb{H}^{(n)}.
\]
If each subsystem is in a state \(\ket{\psi_i} \in \mathbb{H}^{(i)}\), then the global state of the composite system is the product state
\[
\ket{\psi_1} \otimes \ket{\psi_2} \otimes \cdots \otimes \ket{\psi_n}.
\]
\end{postulate}

\section{Simulating Quantum Mechanical Systems}

The primary interest of this work is to explore methods for simulating the time evolution of quantum systems, particularly in the case where the Hamiltonian is time-independent. At first glance, this might appear straightforward: Postulate 3 provides a differential equation with an analytical solution that describes the system's evolution. However, let us consider the size of the relevant Hamiltonian. If we are interested in the evolution of a single electron's spin under a magnetic field, the electron can be represented by a two-dimensional complex vector. In this case, the Hamiltonian is a \(2 \times 2\) complex matrix, simple enough. But what if we want to simulate a system of \(n\) electrons? According to Postulate 4, the quantum state of the full system is the tensor product of the \(n\) individual states. This results in a state vector of dimension \(2^n\) and a Hamiltonian represented by a \(2^n \times 2^n\) matrix.
\\\\
Even though this evolution has an analytical solution, computing \(e^{-i  Ht}\) using a Taylor expansion or other classical numerical methods becomes intractable for large \(n\), due to the exponential growth in matrix size. This quickly exceeds the capacity of classical computers. Faced with this challenge, in 1985 Richard Feynman proposed a radical idea: instead of simulating quantum systems using classical computers, why not use quantum systems themselves? That is, use binary degrees of freedom (such as the spin of an electron) and apply unitary operations (as described by Postulate 3') to approximate the continuous time evolution dictated by Postulate 3. This idea led to the introduction of the quantum bit, or qubit, a quantum analogue of the classical bit.

\begin{definition}
A qubit is a quantum system whose state space is a two-dimensional Hilbert space \(\mathbb{H} \cong \mathbb{C}^2\), called the qubit space. An orthonormal basis of \(\mathbb{H}\), referred to as the computational basis, is given by
\[
\ket{0} = \begin{bmatrix} 1 \\ 0 \end{bmatrix}, \quad 
\ket{1} = \begin{bmatrix} 0 \\ 1 \end{bmatrix}.
\]
The Pauli \(Z\) matrix is associated with the qubit as a standard observable.
\end{definition}

How do we operate on these newly defined qubits? Postulate 3 provides the answer: only unitary transformations are allowed. In contrast to classical computation, which relies on logic gates such as AND, OR and NOT, quantum computation utilizes quantum gates: unitary operations that evolve quantum states in accordance with the principles of quantum mechanics.

\begin{definition}
A quantum gate acting on $n$ qubits is a unitary map
\[
U : \mathbb{H}^{\otimes n} \cong \mathbb{C}^{2^n} \to \mathbb{H}^{\otimes n} \cong \mathbb{C}^{2^n}.
\]
\end{definition}
What does a quantum algorithm look like? It is simply a sequence of quantum gates that evolves an initial quantum state into a final state which, ideally, encodes the solution to a given problem. The goal is that upon measurement, this solution is obtained with high probability. Such quantum algorithms are commonly represented as quantum circuits.

\begin{definition}
    Let \( n, L \in \mathbb{N} \) and \( \{U_1, \ldots, U_n \}\) be a set of quantum gates. We call
\[
U = U_L \ldots U_1 ,
\]
a quantum circuit constructed from the gates \( U_1, \ldots, U_L \) and \( L \in \mathbb{N} \) the length or depth of the circuit relative to the gate set \(\{U_1, \ldots, U_n\}\).
\end{definition}

Let us illustrate this point through concrete examples. Just as classical computation relies on a finite set of universal logic gates, it is neither realistic nor physically feasible to assume that a quantum computer can implement arbitrary unitary transformations directly. Instead, quantum algorithms are built from a small set of elementary gates that are easier to implement in practice. For the purposes of this discussion, we assume access to the following fundamental gates: the Hadamard gate, the \(R_z(\theta)\) gate with fixed angle \(\theta = 2\) and the controlled-NOT (CNOT) gate.

\[
\text{Had} = \frac{1}{\sqrt{2}} \begin{bmatrix}
1 & 1 \\
1 & -1
\end{bmatrix},
\quad \text{(Hadamard gate)}
\]

\[
R_z = \begin{bmatrix}
e^{-i} & 0 \\
0 & e^{i}
\end{bmatrix},
\quad \text{(Z-rotation gate with } \theta = 2\text{)}
\]

\[
\text{CNOT} = \begin{bmatrix}
1 & 0 & 0 & 0 \\
0 & 1 & 0 & 0 \\
0 & 0 & 0 & 1 \\
0 & 0 & 1 & 0
\end{bmatrix},
\quad \text{(Controlled-NOT gate)}
\]
\begin{observation}
    The matrix representation of the CNOT gate provided earlier is expressed with respect to the computational basis for two qubits. If \(\{\ket{0}, \ket{1}\}\) is the computational basis for the single-qubit space \(\mathbb{C}^2\), then the basis for the two-qubit space \(\mathbb{C}^2 \otimes \mathbb{C}^2 \cong \mathbb{C}^4\) is given by the Kronecker products:
\[
\begin{aligned}
\ket{00} &:= \ket{0}\ket{0} := \ket{0} \otimes \ket{0} = \begin{bmatrix}1 \\ 0\end{bmatrix} \otimes \begin{bmatrix}1 \\ 0\end{bmatrix} = \begin{bmatrix}1 \\ 0 \\ 0 \\ 0\end{bmatrix}, \\
\ket{01} &:= \ket{0}\ket{1} := \ket{0} \otimes \ket{1} = \begin{bmatrix}1 \\ 0\end{bmatrix} \otimes \begin{bmatrix}0 \\ 1\end{bmatrix} = \begin{bmatrix}0 \\ 1 \\ 0 \\ 0\end{bmatrix}, \\
\ket{10} &:= \ket{1}\ket{0} := \ket{1} \otimes \ket{0} = \begin{bmatrix}0 \\ 1\end{bmatrix} \otimes \begin{bmatrix}1 \\ 0\end{bmatrix} = \begin{bmatrix}0 \\ 0 \\ 1 \\ 0\end{bmatrix}, \\
\ket{11} &:= \ket{1}\ket{1} := \ket{1} \otimes \ket{1} = \begin{bmatrix}0 \\ 1\end{bmatrix} \otimes \begin{bmatrix}0 \\ 1\end{bmatrix} = \begin{bmatrix}0 \\ 0 \\ 0 \\ 1\end{bmatrix}.
\end{aligned}
\]
Notice that this is precisely the standard basis of \(\mathbb{C}^4\); more generally, the computational basis for a system of \(n\) qubits corresponds to the standard basis of \(\mathbb{C}^{2^n}\).
\end{observation}
In the examples that follow, we consider specific quantum systems evolving under physically meaningful Hamiltonians and show how these evolutions can be simulated using only the gates above for one second. For instance, the Hamiltonian of an electron subjected to a uniform magnetic field in the \(z\)-direction is simply \(H = Z\) (it is not common for a Hermitian matrix to have physical meaning both as an observable and as a Hamiltonian, but the Pauli matrices are a special case).
\begin{example}[$H=Z$]
Given an initial state \( \ket{\psi_0} \), we aim to find a quantum circuit \( U(t) \) such that \( U(t) \ket{\psi_0} \) approximates \( e^{-iZ t} \ket{\psi_0} \) for $t=1$. We can compute this time evolution map as follows:

\[
e^{-iZ t} = e^{-it\ket{0}\bra{0} + it\ket{1}\bra{1}}
= e^{-it} \ket{0}\bra{0} + e^{it} \ket{1}\bra{1}
= \begin{bmatrix}
    e^{-it} & 0 \\
    0 & e^{it}
\end{bmatrix}.
\]
Thus, for \( t = 1 \), we can use \( e^{-i Z t} = R_z \) and therefore:
\[
e^{-i Z} \ket{\psi_0} = R_z \ket{\psi_0},
\] for all \( \ket{\psi_0} \).
\end{example}
Interestingly, consider what happens when applying the previous evolution to the state \( \ket{0} \). That is, how does an electron with spin along the \( z \)-axis evolve under a magnetic field along the \( z \)-axis?

\[
R_z \ket{0} =
\begin{bmatrix}
e^{-i} & 0 \\
0 & e^{i}
\end{bmatrix}
\begin{bmatrix}
1 \\
0
\end{bmatrix}
=
\begin{bmatrix}
e^{-i} \\
0
\end{bmatrix}
= e^{-i} \ket{0}.
\]
This results in a global phase factor \( e^{-i} \), which has no physical effect on measurement outcomes. That is, if the spin is already aligned with the \( z \)-axis, under a magnetic field along \( z \), it experiences no torque, so the spin does not change. Meanwhile, consider an electron with spin aligned along the \( x \)-axis. The state aligned along \( x \) is \( \ket{+} := \frac{1}{\sqrt{2}} \left( \ket{0} + \ket{1} \right) \) and applying the \( R_z \) gate results in:

\[
R_z \ket{+} =
\begin{bmatrix}
e^{-i} & 0 \\
0 & e^{i}
\end{bmatrix}
\cdot \frac{1}{\sqrt{2}} \begin{bmatrix} 1 \\ 1 \end{bmatrix}
=
\frac{1}{\sqrt{2}} \begin{bmatrix}
e^{-i} \\
e^{i}
\end{bmatrix}
= \frac{e^{-i}}{\sqrt{2}} \begin{bmatrix}
1 \\
e^{2i}
\end{bmatrix}.
\]
Here, the relative phase between \( \ket{0} \) and \( \ket{1} \) is modified, meaning that this transformation has an observable effect. The magnetic field is rotating the electron's spin around the \( xy \)-plane. \\\\ So, to simulate the transformation according to the Hamiltonian \( H = Z \), we only needed one quantum gate. But what about for \( H = X \)? That is, a uniform (time-independent) magnetic field in the \( x \)-direction. 
\begin{example}[$H=X$]
Notice that 
\[
X = \text{Had} \circ Z \circ \text{Had}.
\]
We can therefore write:

\[
e^{-i X t} = e^{-i \text{Had} \circ Z \circ\text{Had} \, t}.
\]
Since \( \text{Had}^2 = I \), using the series expansion for the exponential, we get:
\[
e^{-i X t} = \sum_{j=0}^{\infty} \frac{(-i t)^j}{j!} \left( \text{Had} \circ Z \circ \text{Had} \right)^j
= \text{Had} \circ \left( \sum_{j=0}^{\infty} \frac{(-i t)^j}{j!} Z^j \right) \circ \text{Had}
= \text{Had} \circ e^{-i Z t} \circ \text{Had}.
\]
Thus, in this case, to simulate the time evolution of any quantum system according to $H=X$ for $t=1$, we needed three quantum gates: Had, \( R_z \) and Had again.
\end{example}
\begin{example}[$H = Z \otimes Z$]
We now consider a Hamiltonian of the form \( H = Z \otimes Z \). Notice that:
\[
e^{-iZ \otimes Z t} \neq e^{-i Z t} \otimes e^{-i Z t}.
\]
Once again, we can compute the exponential using its series expansion and the fact that \( (Z \otimes Z)^2 = I \):
\[
e^{-i Z \otimes Z t} = \sum_{j=0}^{\infty} \frac{(-i t)^j}{j!} (Z \otimes Z)^j
= \sum_{j=0}^{\infty} \frac{(-i t)^{2j}}{(2j)!} I + \sum_{j=0}^{\infty} \frac{(-i t)^{2j+1}}{(2j+1)!} (Z \otimes Z)^{2j+1}.
\]
This simplifies to:
\[
e^{-i Z \otimes Z t} = \cos(t) I - i \sin(t) Z \otimes Z.
\]

To understand the action of this map, consider its effect on the computational basis states:
\begin{enumerate}

\item For \( \ket{00} \):
\[
e^{-i Z \otimes Z t} \ket{00} = \cos(t) I \ket{00} - i \sin(t) Z \ket{0} \otimes Z \ket{0} = (\cos(t) - i \sin(t)) \ket{00}.
\]

\item For \( \ket{01} \):
\[
e^{-i Z \otimes Z t} \ket{01} = \cos(t) I \ket{01} - i \sin(t) Z \ket{0} \otimes Z \ket{1} = (\cos(t) + i \sin(t)) \ket{01}.
\]

\item For \( \ket{10} \):
\[
e^{-i Z \otimes Z t} \ket{10} = \cos(t) I \ket{10} - i \sin(t) Z \ket{1} \otimes Z \ket{0} = (\cos(t) + i \sin(t)) \ket{10}.
\]

\item For \( \ket{11} \):
\[
e^{-i Z \otimes Z t} \ket{11} = \cos(t) I \ket{11} - i \sin(t) Z \ket{1} \otimes Z \ket{1} = (\cos(t) - i \sin(t)) \ket{11}.
\]
 
\end{enumerate}

Thus, in matrix form:

\[
e^{-i Z \otimes Z t} = \begin{bmatrix}
e^{-it} & 0 & 0 & 0 \\
0 & e^{it} & 0 & 0 \\
0 & 0 & e^{it} & 0 \\
0 & 0 & 0 & e^{-it}
\end{bmatrix}.
\]
In general, for any computational basis state \( \ket{ab} \), this can be written as:

\[
e^{-i Z \otimes Z t} \ket{ab} = e^{-it (-1)^{a \oplus b}} \ket{ab},
\]
where \( \oplus \) denotes addition modulo 2. This transformation can be implemented for $t=1$ in a quantum circuit as follows. Applying a CNOT gate gives:

\[
CNOT \ket{ab} = \ket{a} \ket{a \oplus b}.
\]
Then, applying the \( R_z \) gate to the second qubit yields:

\[
(I \otimes R_z) \, \text{CNOT} \ket{ab} = e^{-i (-1)^{a \oplus b}} \ket{a} \ket{a \oplus b}.
\]
Finally, applying CNOT again:
\[
\text{CNOT} (e^{-i (-1)^{a \oplus b}} \ket{a} \ket{a \oplus b}) = e^{-i (-1)^{a \oplus b}} \ket{a} \ket{a \oplus a \oplus b}= e^{-i (-1)^{a \oplus b}} \ket{a} \ket{b}.
\]
Thus, the full circuit \( \text{CNOT} \circ (I \otimes R_z) \circ \text{CNOT} \) reproduces the evolution under \( e^{-i Z \otimes Z t} \) for one second and requires exactly three quantum gates.
\end{example}
\begin{example}[$H = X \otimes Z$]
We now consider the Hamiltonian \( H = X \otimes Z \). To simulate the corresponding unitary evolution, we begin by observing the identity:
\[
e^{-i X \otimes Z t} = e^{-i (Had \circ Z \circ Had) \otimes Z t} = e^{-i (Had \otimes I)(Z \otimes Z)(Had \otimes I)t}.
\]
Using that $(Had \otimes I)^2 = Had^2 \otimes I^2 = I \otimes I$, we obtain:
\[
e^{-i X \otimes Z t} = (Had \otimes I) \, e^{-i Z \otimes Z t} \, (Had \otimes I).
\]
We have previously shown that:
\[
e^{-i Z \otimes Z } = \text{CNOT} \circ (I \otimes R_z) \circ \text{CNOT}.
\]
Substituting this into the expression for \( e^{-i X \otimes Z t} \) for $t=1$, we find:
\[
e^{-i X \otimes Z } = (Had \otimes I) \circ \text{CNOT} \circ (I \otimes R_z) \circ \text{CNOT} \circ (Had \otimes I).
\]

Thus, the unitary evolution generated by the Hamiltonian \( X \otimes Z \) can be implemented with a quantum circuit consisting of five gates.
\end{example}
What about the Hamiltonian \( H = Z + X \), which represents a uniform magnetic field pointing diagonally in the \( xz \)-plane? Since \( Z \) and \( X \) do not commute, we cannot simply write
\[
e^{-i(Z+X)t} = e^{-i Z t} e^{-i X t}.
\]

Until now, we have relied on case-specific constructions to build quantum circuits that simulate the time evolution generated by particular Hamiltonians. This naturally raises the question: given an arbitrary Hamiltonian, is there a general method to efficiently construct a quantum circuit (composed of elementary gates from a fixed, universal set) that approximates the time evolution it generates? More precisely, how many such gates are required to approximate the unitary \( e^{-iHt} \) up to a given precision and how does this number scale with the size of the system?
\\\\
This question becomes especially important when dealing with Hamiltonians that arise in models of physical interest. One such example is the Ising model, which describes a system of \( n \) spins, each subject to: \begin{enumerate}
\item a uniform magnetic field along the \( x \)-axis (corresponding to the \( X \)-term) and
\item an effective local field along the \( z \)-axis arising from interactions with neighboring spins (captured by the \( Z_i Z_j \)-terms). \end{enumerate}

The Hamiltonian for the transverse-field Ising model on \(n\) qubits is given by
\[
H = -J \sum_{\langle i, j \rangle} Z_i Z_j - h \sum_{i} X_i,
\]
where \( J \in \mathbb{R} \) is the interaction strength and \( h \in \mathbb{R} \) controls the strength of the transverse magnetic field. In this context, we consider a system of \(n\) particles arranged linearly. The notation \( \langle i, j \rangle \) refers to nearest-neighbor pairs, meaning \( j = i + 1 \). That is, the sum over \( \langle i, j \rangle \) runs over all adjacent sites in the 1D chain: \( (1,2), (2,3), \dots, (n-1, n) \). Each term \( Z_i Z_j \) denotes the tensor product of \(n\) matrices acting on the total space \( (\mathbb{C}^2)^{\otimes n} \), where identity operators \(I\) are applied on all qubits except positions \(i\) and \(j\), where \(Z\) acts. Similarly, \(X_i\) denotes the tensor product with an \(X\) acting only on the \(i\)-th qubit and identities elsewhere. This defines a family of Hamiltonians whose size and complexity grow with the number of spins \( n \). For example, for $n=3$, the Hamiltonian is: \[
H = -J \left( Z \otimes Z \otimes I + I \otimes Z \otimes Z \right) - h \left( X \otimes I \otimes I + I \otimes X \otimes I + I \otimes I \otimes X \right),
\] and for $n=4$:\begin{align*}
H = & -J \left( Z \otimes Z \otimes I \otimes I + I \otimes Z \otimes Z \otimes I + I \otimes I \otimes Z \otimes Z \right) \\
    & - h \left( X \otimes I \otimes I \otimes I + I \otimes X \otimes I \otimes I + I \otimes I \otimes X \otimes I + I \otimes I \otimes I \otimes X \right).
\end{align*}
Another important example is the Heisenberg model, described by the Hamiltonian:
\[
H = J \sum_{\langle i, j \rangle} \left( X_i X_j + Y_i Y_j + Z_i Z_j \right),
\]
in which neighboring spins interact via all three components of spin, modeling isotropic spin-spin coupling.
\\\\
More generally, one often encounters Hamiltonians that describe physical systems other than spin chains. For example, the Bose-Hubbard model
\[
H = -t \sum_{\langle i,j \rangle} \left( b_i^\dagger b_j + \text{h.c.} \right) + \frac{U}{2} \sum_i n_i(n_i - 1),
\]
describes bosons hopping on a lattice with on-site repulsion. Here, \( b_i^\dagger \) and \( b_j \) are bosonic creation and annihilation operators at sites \( i \) and \( j \), respectively, and \( n_i = b_i^\dagger b_i \) is the number operator counting bosons at site \( i \). The parameter \( t \) is the hopping amplitude between neighboring sites \( \langle i,j \rangle \), while \( U \) quantifies the strength of the repulsive interaction when two bosons occupy the same site. The abbreviation \text{h.c.} stands for the Hermitian conjugate of the preceding term, which is added to ensure that the Hamiltonian is Hermitian.
\\\\
In all these cases, the central computational challenge is the same: given a Hamiltonian \( H \), how large (in terms of gate count and depth) is the most efficient quantum circuit that simulates the time evolution \( e^{-iHt} \)? How does this complexity scale with the simulation time, precision or number of particles involved? These are precisely the kinds of questions addressed in the field of Hamiltonian Simulation, which forms the core subject of this work.

\pagebreak
\chapter{Hamiltonian Simulation}
\section{Introduction}
Following the discussion from the previous chapter, we begin by formalizing the problem of simulating a quantum system using a quantum computer.
\begin{problem}[Hamiltonian Simulation]
    Given a Hamiltonian \( H \in M_{2^n}(\mathbb{C}) \) and an evolution time \( t > 0 \), find a quantum circuit \( U_{\text{sim}} \) that, for any initial state \( \ket{\psi(0)} \in \mathbb{C}^{2^n} \), produces an approximation \( \ket{\psi_{\text{sim}}} \) of the final state \( \ket{\psi(t)} = e^{-iHt} \ket{\psi(0)} \) up to a prescribed error tolerance \( \epsilon > 0 \).
\end{problem}

\begin{observation}
The worst-case \( \ell_2 \) error in quantum simulation satisfies
    \[
    \max_{\ket{\psi(0)}} \|\ket{\psi(t)} - \ket{\psi_{\text{sim}}} \| = \max_{\ket{\psi(0)}} \| e^{-iHt}\ket{\psi(0)} - U_{\text{sim}}\ket{\psi(0)} \| = \| e^{-iHt} - U_{\text{sim}} \|.
    \]
    This follows because the matrix norm \( \| \cdot \| \) represents the largest possible deviation over all input states, which corresponds to the worst-case state error.

\end{observation}
Clearly, if one had access to arbitrary quantum gates acting on an arbitrary number of qubits, this problem would be trivial: one could simply choose \( U_{\text{sim}} = e^{-iHt} \). However, in practice, for quantum computing, one typically only has access to a finite ``universal" set of one- and two-qubit gates. Universal means that any unitary map \( U \) can be implemented with a circuit consisting of these gates (though not necessarily efficiently). From this point onward, we assume that we are working with such a universal gate set, which transforms the problem into approximating \( e^{-iHt} \) using these gates efficiently. \\\\
To define what we mean by efficient, we first need to establish a notion of computational complexity.

\begin{definition} \label{complexity}
 Let \( f, g: \mathbb{R}^n \to \mathbb{R}^+ \). We say that \( f(\mathbf{x}) = O(g(\mathbf{x})) \) if there exist positive constants \( c >0\) and \( \mathbf{x}_0 \in \mathbb{R}^n \) such that for all \( \mathbf{x} \in \mathbb{R}^n \) with \( \|\mathbf{x}\| \geq \|\mathbf{x}_0\| \), the following inequality holds:
    \[
    f(\mathbf{x}) \leq c \, g(\mathbf{x}).
    \]
    We say that \( g(\mathbf{x}) \) is an asymptotic upper bound for \( f(\mathbf{x}) \).
\end{definition}
In particular, we are interested in how the complexity of the circuit used to approximate \( e^{-iHt} \) scales with respect to the evolution time \( t \), the size of the system \( n \) and the required precision \( \epsilon \).
\begin{definition}
    Let \( f: \mathbb{R}^n \to \mathbb{R}^+ \) be a function that depends on a set of variables \( x_1, x_2, \dots, x_n \). We say that \( f(x_1, x_2, \dots, x_n) = O(\text{poly}(x_1, x_2, \dots, x_n)) \) if there exist constants \( k_1, k_2, \dots, k_n \geq 0 \) such that:
    \[
    f(x_1, x_2, \dots, x_n) = O(x_1^{k_1} x_2^{k_2} \cdots x_n^{k_n}).
    \]
    That is, \( f \) grows at most polynomially in the variables \( x_1, x_2, \dots, x_n \), with the degree of each \( x_i \) determined by \( k_i \). Similarly, we say that \( f(x_1, x_2, \dots, x_n) = O(\text{polylog}(x_1, x_2, \dots, x_n)) \) if there exist constants \( k_1, k_2, \dots, k_n \geq 0 \) such that:
    \[
    f(x_1, x_2, \dots, x_n) = O\left((\log x_1)^{k_1} (\log x_2)^{k_2} \cdots (\log x_n)^{k_n}\right),
    \] meaning that \( f \) grows at most polylogarithmically in its inputs.
\end{definition}

\begin{definition}
    A Hamiltonian \( H \in M_{2^n}(\mathbb{C}) \) can be efficiently simulated if, for any \( t > 0 \) and \( \epsilon > 0 \), there exists a quantum circuit \( U \) consisting of \( O\left(\text{poly}(n, t, 1/\epsilon)\right) \) gates such that
    \[
    \| U - e^{-iHt} \| < \epsilon.
    \]
\end{definition}
We aim to understand the conditions under which a Hamiltonian can be efficiently simulated. We will see that it is generally not possible to efficiently simulate arbitrary Hamiltonians. Instead, we will describe several classes of Hamiltonians that can be efficiently simulated. \\\\
One such class consists of Hamiltonians that act nontrivially on only a constant number of qubits.

\begin{definition}
We say that a Hamiltonian \( H  \in M_{2^n}(\mathbb{C}) \) acts non-trivially on at most \( k < n \) qubits if \( H \) is the tensor product of a matrix \( H' \) acting on \( k \) qubits and the identity matrix acting on the remaining \( n - k \) qubits. That is,
    \[
    H = I^{\otimes (n-k)} \otimes H',
    \]
    up to a reordering of the qubits.
\end{definition}
\begin{proposition}
    Let \( H \in M_{2^n}(\mathbb{C})\) be a Hamiltonian that acts non-trivially on at most \( k \) qubits. Then:
    \[
    e^{-iHt} = e^{-i\left(I^{\otimes (n-k)} \otimes H'\right)t} = I^{\otimes (n-k)} \otimes e^{-iH't}.
    \]
\end{proposition}
\begin{proof}
       The matrix exponential can be computed via its power series expansion:
    \begin{align*}
        e^{-iHt} = \sum_{m=0}^\infty \frac{(-iHt)^m}{m!} &=  \sum_{m=0}^\infty \frac{(-i)^m t^m \left(I^{\otimes (n-k)} \otimes H'\right)^m}{m!} \\ &= \sum_{m=0}^\infty \frac{(-i)^m t^m \left(I^{\otimes (n-k)}\right)^m \otimes (H')^m}{m!} \\ &= \sum_{m=0}^\infty \frac{(-i)^m t^m I^{\otimes (n-k)} \otimes (H')^m}{m!} \\ &= I^{\otimes (n-k)} \otimes \left( \sum_{m=0}^\infty \frac{(-iH't)^m}{m!} \right) \\ &=  I^{\otimes (n-k)} \otimes e^{-iH't}.
    \end{align*}
\end{proof}
We can therefore interpret the corresponding unitary transformation as acting non-trivially on at most \( k \) qubits. The fact that these Hamiltonians can be efficiently simulated follows from the well-known Solovay-Kitaev theorem. This theorem will not be proven here, as it is beyond the scope of this work. A discussion of the theorem can be found in \cite{Dawson2006}.

\begin{theorem}[Solovay-Kitaev]
    Let \( U \) be a unitary map that acts non-trivially on \( k = O(1) \) qubits and let \( S \) be an arbitrary finite universal set of one- and two-qubit quantum gates. Then \( U \) can be approximated in the spectral norm to within \( \epsilon \) using \( O(\log^c(1/\epsilon)) \) gates from \( S \), for some constant \( c < 4 \).
\end{theorem}
The key idea of the proof is that any one- and two-qubit gate can be efficiently approximated using \( O(\log^c(1/\epsilon)) \) gates from a finite universal set of one- and two-qubit gates. However, approximating an arbitrary gate acting on \( n \) qubits with only one- and two-qubit gates incurs an exponential overhead in \( n \). Nevertheless, since \( H \) acts non-trivially on at most \( k \) qubits, independent of the size \( n \) of the system, the overall complexity remains independent of \( n \). Specifically, the complexity is \( O(2^k \log^c(1/\epsilon)) = O(\log^c(1/\epsilon)) \). \\\\
It is uncommon for a Hamiltonian to act non-trivially on at most a constant number of qubits. However, it is very common for a Hamiltonian to be a sum of Hamiltonians that each act non-trivially on at most a constant number of qubits. Prominent examples include the Ising and Heisenberg models, where each interaction term involves at most two qubits. Such local Hamiltonians arise naturally in many physical systems, as interactions typically weaken with increasing spatial separation or energy difference between subsystems. 
\begin{definition}
    A Hamiltonian \( H \in M_{2^n}(\mathbb{C}) \) is said to be \textit{k-local} if it can be written as a sum
    \[
    H = \sum_{j=1}^{L} H_j,
    \]
    for some \( L \), where each \( H_j \) is a Hamiltonian that acts non-trivially on at most \( k \) qubits.
\end{definition}
The question now is: given \( H = \sum_{j=1}^{L} H_j \), how can we efficiently simulate \( e^{-iHt} \) assuming that we can efficiently simulate \( e^{-iH_j t} \) for each \( H_j \)? The next results show that if all of the \( H_j \)'s commute, then this process becomes relatively straightforward.

\begin{proposition}
Let \( H = \sum_{j=1}^{L} H_j \), where \( H_j \) are Hamiltonians and suppose that \( [H_k, H_l] = 0 \) for all \( k, l \). Then, for all \( t \), we have the following factorization:

\[
e^{-iHt} = e^{-iH_1 t} e^{-iH_2 t} \cdots e^{-iH_L t}.
\]
\end{proposition}

So, for commuting \( H_j \), the quantum circuit to approximate \( e^{-iHt} \) is simply the composition of the quantum circuits that approximate each \( e^{-iH_j t} \). To analyze the complexity of the overall circuit, we need the following result. It essentially says that errors in the approximation of one quantum circuit accumulate at most linearly.

\begin{lemma}\label{linear_error}
    Let \( U_i, V_i \) be unitary matrices satisfying \( \|U_i - V_i\| \leq \epsilon \) for all \( i \in \{1, 2, \dots, L\} \). Then
    \[
        \| U_L \dots U_2 U_1 - V_L \dots V_2 V_1 \| \leq L \epsilon.
    \]
\end{lemma}

\begin{proof}
We use induction on \( L \). For \( L = 1 \), the lemma is trivial. Now, suppose the lemma holds for a particular value of \( L-1 \). By the triangle inequality, we obtain:
\begin{eqnarray*}
&& \left\|U_L U_{L-1} \dots U_1 - V_L V_{L-1} \dots V_1\right\| \\
&&\qquad = \left\|U_LU_{L-1} \dots U_1 - U_L V_{L-1} \dots V_1 + U_LV_{L-1} \dots V_1 - V_L V_{L-1} \dots V_1\right\| \\
&&\qquad \leq \left\|U_LU_{L-1} \dots U_1 - U_L V_{L-1} \dots V_1\right\| + \left\|U_LV_{L-1} \dots V_1 - V_L V_{L-1} \dots V_1\right\| \\
&&\qquad = \left\|U_L (U_{L-1} \dots U_1 - V_{L-1} \dots V_1)\right\| + \left\|(U_L - V_L) V_{L-1} \dots V_1\right\| \\
&&\qquad = \left\| U_{L-1} \dots U_1 - V_{L-1} \dots V_1 \right\| + \left\| U_L - V_L \right\| \\
&&\qquad \leq (L-1)\epsilon + \epsilon = L\epsilon.
\end{eqnarray*}

Here, we used that the norm is invariant under unitary transformations in the second-to-last equality and the inductive hypothesis in the final inequality.
\end{proof}

It follows now that, since each \( H_j \) can be simulated with \( O(\log^c(1/\epsilon)) \) gates, we can simulate each of them with precision \( \epsilon / L \) and therefore, by the lemma, simulate \( e^{-iHt} \) using \( O(L \log^c(L/\epsilon)) \) gates. Notice that the complexity doesn't depend on \( t \). 
What about the dependence on \( n \)? Well, clearly \( L \) depends on \( n \). Since \( M_{2^n}(\mathbb{C}) \) has dimension \( 2^{2n} \), \( H \) will generally be written as \( H = \sum_{j=1}^{4^n} H_j \). That is, even if each \( H_j \) could be simulated with \( O\left(\text{poly}(n, t, 1/\epsilon)\right) \) gates, a general Hamiltonian will be expressed as an exponential in \( n \) number of them and therefore not be overall efficiently simulable. 
We conclude that arbitrary Hamiltonians cannot be efficiently simulated (or more precisely, at least not with this method). Instead, we turn our attention to the class of naturally occurring \( k \)-local Hamiltonians.

\begin{remark}
To study how the complexity of \( k \)-local Hamiltonians scales with the number of qubits \( n \), we must formalize what it means for such Hamiltonians to ``grow'' with \( n \). Specifically, we consider families \( \{H^{(n)}\}_{n\in\mathbb{N}} \) of Hamiltonians, where each \( H^{(n)} \) acts on \( n \) qubits and is a sum of \( k \)-local terms drawn from a fixed finite set of interaction types.
\\\\
For example, in the Ising model, each Hamiltonian \( H^{(n)} \) includes terms of the form \( Z_i Z_j \) for neighboring \( i, j \) and \( X_i \) for each qubit \( i \). Similarly, the Heisenberg model uses fixed local terms \( X_i X_j \), \( Y_i Y_j \) and \( Z_i Z_j \) for neighboring pairs. As \( n \) increases, the number of terms grows, but the types of terms do not change.
\\\\
More generally, we assume that each \( H^{(n)} \) is constructed from a finite, \( n \)-independent set of local interaction types. This is a physically natural restriction: it ensures that the family describes a consistent interaction model, rather than introducing entirely new physics at each system size. Furthermore, it guarantees that the total number of terms remains polynomial in \( n \), assuming locality is preserved. More precisely, let $d_k=O(1)$ denote the number of types of Hamiltonians acting on \( k \) qubits that appear in the sum. Then, 
\begin{align*}
L &\leq \sum_{i=1}^k d_i\binom{n}{i} = \sum_{i=1}^k d_i\frac{n!}{(n-i)!i!} 
\leq \max_{i=1,\dots,k} d_i \sum_{i=1}^k \frac{n(n-1)\cdots(n-i+1)}{i!} \\
&= O\left(\sum_{i=1}^k n^i\right) = O(n^k).
\end{align*}
Notice that we could even allow each $d_k$ to grow polynomially with $n$.
\end{remark}

\begin{theorem}
    Let \( H \in \mathcal{M}_{2^n}(\mathbb{C})\) be a \( k \)-local Hamiltonian that can be written as the sum of \( L= O(n^k) \) commuting terms \( H_j \). Then, for any \( t >0 \), there exists a quantum circuit that approximates \( e^{-iHt} \) to within \( \epsilon \) using 
\[
O(n^k \log^c(n^k/\epsilon)),
\] gates.
\end{theorem}

\begin{proof}
    We have \( H = \sum_{j=1}^{L} H_j \), where each \( H_j \) acts non-trivially on at most \( k \) qubits. Let \( U_j \) be the circuit that simulates \( H_j \) to precision \( \epsilon/L \), requiring \( O(\log^c(L/\epsilon)) \) gates, following the Solovay-Kitaev theorem. Applying Lemma \ref{linear_error}, we find that by composing all these circuits, 
    \[
   \| e^{-iHt}-U_1 U_2 \cdots U_L\| =
    \| e^{-iH_1 t} e^{-iH_2 t} \cdots e^{-iH_L t} - U_1 U_2 \cdots U_L \| \leq \epsilon,
    \]
    using a total of \( O(L \log^c(L/\epsilon)) \) gates. Substituting our worst-case \( L \), we obtain \[ O(n^k \log^c(n^k/\epsilon)). \]
\end{proof}
\section{The Lie-Trotter Product Formula}
Unfortunately, requiring all of the \( H_j \) to commute is too strong a condition. For example, in the Ising model, all of the terms involve Pauli matrices and the Pauli matrices \( X, Y, Z \) do not commute with each other. In fact, they anti-commute pairwise. So how can we proceed when the \( H_j \) do not commute? To address this, we introduce a tool that will simplify our analysis.

\begin{definition}
    Let \( \left( V ,\|\cdot\| \right) \) be a normed vector space. Given functions \( f, g \colon \mathbb{R}^n \to V \) and \( h \colon \mathbb{R}^n \to \mathbb{R}^+ \), we write
    \[
    f(\mathbf{x}) = g(\mathbf{x}) + \mathcal{O}(h(\mathbf{x})),
    \]
    if there exist constants \( C > 0 \) and \( \delta > 0 \) such that for all \( \mathbf{x} \) with \( 0 < \|\mathbf{x}\| < \delta \), we have
    \[
    \|f(\mathbf{x}) - g(\mathbf{x})\| \leq C \cdot h(\mathbf{x}).
    \]
\end{definition}
Notice the resemblance with Definition~\ref{complexity}. We will be using both notations throughout the rest of this text. To avoid confusion, we adopt the convention that \( \mathcal{O} \) will denote error bounds in approximations, while \( O \) will refer to computational complexity. Nevertheless, the distinction should be clear from context: one takes the form \( f = g + \mathcal{O}(h) \) and the other \( f = O(g) \). \\\\ We are now ready to address the simulation of a general \( k \)-local Hamiltonian. The core of many quantum simulation algorithms lies in the following asymptotic approximation theorem.

\begin{theorem}
    Let $A, B \in M_{2^n}(\mathbb{C})$ be Hamiltonians. Then, for all $t \in \mathbb{R}$, we have the following limit:
    \begin{equation*}
        \lim_{m \to \infty} \left( e^{-i A t / m} e^{-i B t / m} \right)^m = e^{-i(A + B) t}.
     \end{equation*}
\end{theorem}

\begin{proof}
    We begin by expanding $e^{-i A t / m}$ and $e^{-i B t / m}$:
    \begin{align*}
        e^{-i A t / m} e^{-i B t / m} &= \left( I + \frac{-i A t}{m} + \mathcal{O}\left(\frac{1}{m^2}\right) \right) \left( I + \frac{-i B t}{m} + \mathcal{O}\left(\frac{1}{m^2}\right) \right) \\
        &= I + \frac{-i (A + B) t}{m} + \mathcal{O}\left(\frac{1}{m^2}\right).
    \end{align*}

    Now, we raise this expression to the power $m$ and expand using the binomial series:
    \begin{align*}
        \left( e^{-i A t / m} e^{-i B t / m} \right)^m &= \left( I + \frac{-i (A + B) t}{m} + \mathcal{O}\left( \frac{1}{m^2} \right) \right)^m \\
        &= \sum_{k=0}^m \binom{m}{k} \left( \frac{-i (A + B) t}{m} + \mathcal{O}\left(\frac{1}{m^2}\right)\right)^k \\
        &= \sum_{k=0}^{m} \binom{m}{k} \left(\left( \frac{-i (A + B) t}{m} \right)^k + \mathcal{O}\left( \frac{1}{m^{k+1}} \right)\right) \\
        &= \sum_{k=0}^{m} \binom{m}{k} \frac{1}{m^k}\left( -i \left(A + B\right) t \right)^k+ \mathcal{O}\left(\frac{1}{m} \sum_{k=0}^{m} \binom{m}{k}\frac{1}{m^{k}} \right) \\ 
        &= \sum_{k=0}^{m} \binom{m}{k} \frac{1}{m^k}\left( -i \left(A + B\right) t \right)^k+ \mathcal{O}\left(\frac{1}{m}\left( 1 + \frac{1}{m}\right)^m\right) \\
        &= \sum_{k=0}^{m} \binom{m}{k} \frac{1}{m^k}\left( -i \left(A + B\right) t \right)^k+ \mathcal{O}\left( \frac{1}{m} \right) . \tag{*} \label{trottereq}
    \end{align*}

    Next, we analyze the binomial coefficient:
    \begin{align*}
        \binom{m}{k} \frac{1}{m^k} = \frac{m!}{k! (m-k)!} \frac{1}{m^k} &= \frac{m(m-1)\cdots (m-k+1)}{k!} \frac{1}{m^k} \\
        &= \frac{1}{k!}\frac{m-1}{m}\cdots\frac{m-k+1}{m} \\
        &= \frac{1}{k!} \left( 1 - \frac{1}{m} \right) \left( 1 - \frac{2}{m} \right) \cdots \left( 1 - \frac{k-1}{m} \right) \\ 
        &= \frac{1}{k!} \left( 1 + \mathcal{O}\left(\frac{1}{m}\right) \right).
    \end{align*}
    Substituting this back in \ref{trottereq}, we get:
    \begin{align*}
     \left( e^{-i A t / m} e^{-i B t / m} \right)^m &= \sum_{k=0}^{m} \frac{\left( -i (A + B) t \right)^k}{k!}  \left( 1 + \mathcal{O}\left( \frac{1}{m} \right) \right) + \mathcal{O}\left( \frac{1}{m} \right) \\
    &= \sum_{k=0}^{m} \frac{(-i (A + B) t)^k}{k!} + \mathcal{O}\left( \frac{1}{m} \right) \sum_{k=0}^{m} \frac{(-i (A + B) t)^k}{k!} +  \mathcal{O}\left( \frac{1}{m} \right).
    \end{align*}

     Taking the limit as $m \to \infty$, the error term vanishes and we are left with:

    \begin{align*}
        \lim_{m \to \infty} \left( e^{-i A t / m} e^{-i B t / m} \right)^m &= \sum_{k=0}^{\infty} \frac{(-i (A + B) t)^k}{k!} = e^{-i (A + B) t}.
    \end{align*}
\end{proof}
Notice that what we are doing is simulating each \( H_j \) for progressively smaller time intervals and repeating this process multiple times. \\\\Let's generalize this for arbitrary Hamiltonians.
\begin{corollary}[Lie-Trotter Formula]
    Let $H_1,H_2 ...,H_L \in M_{2^n}(\mathbb{C})$ be Hamiltonians. Then, for all $t \in \mathbb{R}$, we have the following limit:
    \begin{equation*}
        \lim_{m \to \infty} \left( e^{-i H_1 t / m} e^{-i H_2 t / m} \cdots  e^{-i H_L t / m} \right)^m = e^{-i(H_1 + H_2 + \cdots + H_L) t}.
     \end{equation*}
\end{corollary}
\begin{proof}
    It follows from the fact that \begin{align*}
        e^{-i H_1 t / m} e^{-i H_2 t / m} \cdots e^{-i H_L t / m} \\  &\hspace{-4.2cm}= \left( I + \frac{-i H_1 t}{m} + \mathcal{O}\left(\frac{1}{m^2}\right) \right) \left( I + \frac{-i H_2 t}{m} + \mathcal{O}\left(\frac{1}{m^2}\right) \right) \cdots \left( I + \frac{-i H_L t}{m} + \mathcal{O}\left(\frac{1}{m^2}\right) \right)\\
        &\hspace{-4.2cm}= I + \frac{-i (H_1 + H_2 + \cdots + H_L) t}{m} + \mathcal{O}\left(\frac{1}{m^2}\right).
    \end{align*}
    The rest of the proof follows analogously from the case with only two matrices.
\end{proof}
A simulation using a finite number of steps can be achieved by choosing a sufficiently large value of \( m \). This truncation introduces some error, which must be kept small. Notice that since \( -i \) is a constant, we can omit it for the purpose of studying the error and reintroduce it at the end. Let's examine this error in more detail.
First, for the case $L=2$ expand \( e^{(A+B)t} \) and \( e^{At} e^{Bt} \) as follows:

\begin{align*} 
e^{(A+B)t} &= I + (A+B)t + \frac{1}{2} (A+B)^2t^2 + \mathcal{O}(t^3) \\
&= I + (A+B)t + \frac{1}{2} \left(A^2 + AB + BA + B^2\right)t^2 + \mathcal{O}(t^3).
\end{align*}
\begin{align*}
e^{At} e^{Bt} &=\left( I + At + \frac{1}{2} A^2t^2 + \mathcal{O}(t^3) \right)
\left( I + Bt + \frac{1}{2} B^2t^2 + \mathcal{O}(t^3) \right) \\
&= I + (A+B)t + \frac{1}{2} \left(A^2 + 2AB + B^2\right)t^2 + \mathcal{O}(t^3).
\end{align*}
Therefore, we conclude that:
\begin{align*}
e^{-i(A+B)t} = e^{-iAt} e^{-iBt} + \mathcal{O}(t^2).
\end{align*}
Similarly, we can write:
\begin{align*}
e^{-i(A+B)t/m} = e^{-iAt/m} e^{-iBt/m} + \mathcal{O}((t/m)^2).
\end{align*}
Finally, applying Lemma \ref{linear_error}, we obtain:
\begin{align*}
e^{-i(A+B)t} = \left(e^{-iAt/m} e^{-iBt/m}\right)^m + \mathcal{O}(t^2/m).
\end{align*}
Now let's prove it for the general case.
\begin{proposition}
    Let $H = \sum_{j=1}^LH_j \in M_{2^n}(\mathbb{C})$ then
    \begin{equation*}
        e^{-iHt}=e^{-iH_1t}\cdots e^{-iH_Lt}  + \mathcal{O}(L^2t^2).
    \end{equation*}
\end{proposition}
\begin{proof}
First, expand to second order:
\begin{align*}
e^{t\sum_j H_j} &= I +t\sum_j H_j + \frac{t^2}{2}\left(\sum_j H_j\right)^2 + \mathcal{O}(t^3), \\
\prod_{j=1}^L e^{t H_j} &= \prod_{j=1}^L \left(I +t H_j + \frac{t^2}{2} H_j^2 + \mathcal{O}(t^3)\right) \\
&= I +t\sum_j H_j + \frac{t^2}{2}\sum_j H_j^2 + t^2\sum_{j<k} H_j H_k + \mathcal{O}(t^3).
\end{align*}
Now, subtract the expansions:
\begin{align*}
e^{t\sum_j H_j} - \prod_{j=1}^L e^{t H_j} &= \frac{t^2}{2}\left( \left(\sum_j H_j\right)^2-\sum_j H_j^2 - 2\sum_{j<k} H_j H_k \right) + \mathcal{O}(t^3) \\ &=\frac{t^2}{2}\left( \sum_{j\neq k} H_jH_k-2\sum_{j<k} H_j H_k  \right) + \mathcal{O}(t^3) \\
&= \frac{t^2}{2}\sum_{j<k} [H_k, H_j] + \mathcal{O}(t^3).
\end{align*}

Since the number of commutator terms is \( \binom{L}{2} = \frac{L(L-1)}{2} \), we obtain:

\[
e^{-it\sum_j H_j} - \prod_{j=1}^L e^{-it H_j} = \mathcal{O}(L^2t^2).\]

\end{proof}
\begin{theorem}
    Let \( H = \sum_{j=1}^L H_j \) be a \( k \)-local Hamiltonian acting on \( n \) qubits. Then, for any \( t > 0 \) and error \( \epsilon > 0 \), there exists a quantum circuit that approximates the unitary \( e^{-iHt} \) to within error \( \epsilon \), using a number of gates
    \[
        O\left( \frac{n^{3k} t^2}{\epsilon} \log^c\left( \frac{n^{3k} t^2}{\epsilon^2} \right) \right),
    \]
    where \( c \) is the constant from the Solovay-Kitaev theorem.
\end{theorem}

\begin{proof}
    From the previous proposition, we have:
    \[
        e^{-it\sum_j H_j} - \prod_{j=1}^L e^{-it H_j} = \mathcal{O}(L^2 t^2).
    \]
    Therefore,
    \[
        e^{-i\frac{t}{m}\sum_j H_j} - \prod_{j=1}^L e^{-i\frac{t}{m} H_j} = \mathcal{O}\left(\frac{L^2 t^2}{m^2}\right).
    \]
    Applying Lemma~\ref{linear_error}, which bounds the error of repeating a small-step approximation \( m \) times:
    \[
        e^{-it\sum_j H_j} - \left(\prod_{j=1}^L e^{-i\frac{t}{m} H_j}\right)^m = \mathcal{O}\left( \frac{L^2 t^2}{m} \right).
    \]
    To ensure the overall simulation error is at most \( \epsilon \), it suffices to choose
    \[
        m = O\left( \frac{L^2 t^2}{\epsilon} \right).
    \]
    The total number of exponentials (i.e., gate blocks of the form \( e^{-i \frac{t}{m} H_j} \)) is then:
    \[
        Lm = O\left( \frac{L^3 t^2}{\epsilon} \right).
    \]
    Since quantum gates come from a finite universal set, we must approximate each \( e^{-i \frac{t}{m} H_j} \) with error at most \( \epsilon / (Lm) \) to guarantee total error at most \( \epsilon \). By the Solovay-Kitaev theorem, each such gate can be implemented using
    \[
        O\left( \log^c\left( \frac{Lm}{\epsilon} \right) \right),
    \]
    elementary gates, for some constant \( c < 4 \).\\
    Therefore, the total gate complexity becomes:
    \[
        O\left( Lm \log^c\left( \frac{Lm}{\epsilon} \right) \right) = O\left( \frac{L^3 t^2}{\epsilon} \log^c\left( \frac{L^3 t^2}{\epsilon^2} \right) \right).
    \]
    Finally, using the assumption that \( H \) is \( k \)-local, the number of terms in the sum satisfies \( L = O(n^k) \), since each \( H_j \) acts on at most \( k \) qubits. Thus,
    \[
        \frac{L^3 t^2}{\epsilon} = O\left( \frac{n^{3k} t^2}{\epsilon} \right),
    \]
    and the total gate count is:
    \[
        O\left( \frac{n^{3k} t^2}{\epsilon} \log^c\left( \frac{n^{3k} t^2}{\epsilon^2} \right) \right).
    \]
\end{proof}

Before continuing, let us make two observations regarding complexity estimates. Every time we apply the Solovay--Kitaev theorem, we obtain a final gate complexity of the form
\[
O\left(G \log^c\left(\frac{G}{\epsilon}\right)\right),
\]
where \( G \) is the number of arbitrary gates needed to approximate \( e^{-iHt} \). To simplify these expressions and improve clarity, we will ignore polylogarithmic factors (terms of the form \(O(\text{polylog}(n, t, \epsilon) \)) from now on. This reflects the true power and usefulness of the Solovay--Kitaev theorem: we can analyze the asymptotic gate complexity \( G \) and when moving to a finite universal gate set, the overhead will only be polylogarithmic in \( G/\epsilon \).
For example, in the theorem above, we obtained a gate complexity of 
\[
O\left(\frac{n^{3k}t^2}{\epsilon} \log^c\left(\frac{n^{3k}t^2}{\epsilon^2}\right)\right).
\]
From now on, we will simplify such expressions to just
\[
O\left(\frac{n^{3k}t^2}{\epsilon}\right),
\]
understanding that the hidden factors include polylogarithmic terms.
Furthermore, for clarity, we will write expressions in terms of \( L \), the number of local terms in the Hamiltonian decomposition and simply assume that \( L = O(n^k) \) for a \( k \)-local Hamiltonian. That is, we will use bounds like \( O(L^3 t^2 / \epsilon) \), knowing that \( L \) depends polynomially on the number of qubits \( n \). \section{Higher-Order Product Formulas}
It seems somewhat undesirable that, in order to simulate a Hamiltonian for time \( t \), the algorithm we have so far exhibits a complexity with quadratic dependence on \( t \), i.e., \( O(t^2) \). This behavior arises from the second-order approximation 
\[
e^{-iH_1t} e^{-iH_2t} \cdots e^{-iH_Lt} \approx e^{-iHt},
\]
whose error is \( \mathcal{O}(L^2t^2) \). In the following sections, we will develop higher-order approximations that reduce this dependence on \( t \), resulting in more efficient simulations for large evolution times.
Notice that, in this initial approximation, the number of local terms \( L \) has the same order as \( t \) in the error bound. This will generally continue to be the case. However, in the arguments that follow, explicitly keeping track of \( L \) in each step (especially when dealing with combinatorial expansions and term-by-term error estimates) would significantly lengthen the proofs and detract from the core ideas.
Therefore, to streamline the presentation and maintain focus, we will temporarily suppress the dependence on \( L \) and concentrate on the scaling with respect to \( t \). Once we derive the optimal approximation scheme, we will return to a rigorous bound on the number of quantum gates required to simulate \( H \), at which point the dependence on \( L \) will reappear with the same order as the time parameter \( t \) (as expected). \\\\ We begin by defining our goal. As before, we omit the factor $-i$ for convenience.\begin{definition}
    Let \( H = \sum_{j=1}^L H_j \) be a Hamiltonian. We say that a product formula
    \[
    S(t) = \prod_{i=1}^{q} e^{p_{i,1} H_1 t} \cdots e^{p_{i,L} H_L t},
    \]
    is an approximation of order k of the time evolution map $e^{Ht}$ if
    \[
    e^{Ht} = S(t) + \mathcal{O}(t^{k+1}).
    \]
\end{definition}

Observe that the expression $\prod_{j=1}^L e^{t H_j}$ constitutes a first-order product formula. To build intuition for constructing higher-order approximations, let us first consider the simple case where $L = 2$. It is straightforward to construct a second-order product formula in this case:

\begin{align*} 
e^{At/2}e^{Bt}e^{At/2} &=\left( I + \frac{At}{2} +\frac{A^2t^2}{8} + \mathcal{O}(t^3) \right)
\left( I + Bt + \frac{B^2t^2}{2} + \mathcal{O}(t^3) \right) \left( I + \frac{At}{2} +\frac{A^2t^2}{8} + \mathcal{O}(t^3) \right) \\
&= \left(I + \left(\frac{A}{2} + B\right)t  + \left(\frac{A^2}{8} + \frac{AB}{2}  + \frac{B^2}{2}\right)t^2+ \mathcal{O}(t^3)\right)\left( I + \frac{At}{2} +\frac{A^2t^2}{8} + \mathcal{O}(t^3) \right) \\
&= I + (A + B)t  + \left( \frac{A^2}{8} + \frac{AB}{2} + \frac{B^2}{2} + \frac{A^2}{4} + \frac{BA}{2} + \frac{A^2}{8} \right)t^2 + \mathcal{O}(t^3) \\
&= I + (A + B)t + \frac{1}{2} \left( A^2 + AB + BA + B^2 \right)t^2 + \mathcal{O}(t^3).
\end{align*}
On the other hand, we have
\[
e^{(A+B)t} = I + (A + B)t + \frac{1}{2} \left( A^2 + AB + BA + B^2 \right)t^2 + \mathcal{O}(t^3).
\]
Thus, we conclude that
\[
e^{(A+B)t} = e^{At/2} e^{Bt} e^{At/2} + \mathcal{O}(t^3),
\]
which confirms that this is a second-order product formula.
Determining the coefficients \( p_{i,j} \) by directly expanding the product formula and matching each term to the corresponding term in the Taylor expansion of the exponential map \( e^{(A+B)t} \) quickly becomes intractable. Fortunately, in 1990, Masuo Suzuki \cite{suzukimonte, suzukigeneral} introduced a clever recursive method for constructing product formulas of arbitrary order.
We denote our second-order product formula as
\[
S_2(t) = e^{At/2} e^{Bt} e^{At/2}.
\]
The idea behind constructing higher-order formulas is to introduce a free parameter \( s \in \mathbb{R} \) that creates redundancy in the expression \( e^{(A+B)t} \), allowing for cancellation of higher-order error terms. Observe the identity:
\[
e^{(A+B)t} = e^{s(A+B)t} e^{(1 - 2s)(A+B)t} e^{s(A+B)t}.
\]
We now approximate each exponential using the second-order formula \( S_2 \), leading to a new composition with a free parameter:
\[
S_3(t) = S_2(st) S_2((1 - 2s)t) S_2(st).
\]
Explicitly, this gives:
\[
S_3(t) = e^{\frac{s}{2} t A} e^{s t B} e^{\frac{1 - s}{2} t A} e^{(1 - 2s) t B} e^{\frac{1 - s}{2} t A} e^{s t B} e^{\frac{s}{2} t A}.
\]

To determine whether \( S_3(t) \) is a third-order approximation, we expand each component using the error expression for \( S_2 \):
\begin{align*}
S_3(t) &= S_2(st) S_2((1 - 2s)t) S_2(st) \\
&= \left(e^{s(A+B)t} + R_3(\{A,B\}) s^3 t^3 + \mathcal{O}(t^4)\right)
\left(e^{(1 - 2s)(A+B)t} + R_3(\{A,B\}) (1 - 2s)^3 t^3 + \mathcal{O}(t^4)\right) \cdot \\& \quad \cdot 
\left(e^{s(A+B)t} + R_3(\{A,B\}) s^3 t^3 + \mathcal{O}(t^4)\right) \\
&= e^{(A+B)t} + \left(2s^3 + (1 - 2s)^3\right) R_3(\{A,B\}) t^3 + \mathcal{O}(t^4),
\end{align*}
where \( R_3(\{A,B\}) \) is a matrix that depends on \( A \), \( B \) and their products, capturing the leading third-order error term in the expansion.
Thus, to eliminate the third-order error term and obtain a formula of order 3, we must choose \( s \) such that
\[
2s^3 + (1 - 2s)^3 = 0.
\]
Solving this gives \( s = \frac{1}{2 - \sqrt[3]{2}} \). With this choice of \( s \), we conclude:
\[
S_3(t) = e^{(A+B)t} + \mathcal{O}(t^4),
\]
so \( S_3(t) \) is indeed a third-order product formula.
Similarly, we can construct a fourth-order product formula by applying the same recursive structure to \( S_3(t) \). Define:
\[
S_4(t) = S_3(st) S_3((1 - 2s)t) S_3(st).
\]
Expanding each component using the known error of \( S_3 \), we have:
\begin{align*}
S_4(t) &= \left(e^{s(A+B)t} + R_4(\{A,B\}) s^4 t^4 + \mathcal{O}(t^5)\right)
         \left(e^{(1 - 2s)(A+B)t} + R_4(\{A,B\}) (1 - 2s)^4 t^4 + \mathcal{O}(t^5)\right) \cdot \\ & \quad \cdot
         \left(e^{s(A+B)t} + R_4(\{A,B\}) s^4 t^4 + \mathcal{O}(t^5)\right) \\
&= e^{(A+B)t} + \left(2s^4 + (1 - 2s)^4\right) R_4(\{A,B\}) t^4 + \mathcal{O}(t^5).
\end{align*}
To eliminate the fourth-order error term and achieve an approximation of order 4, we must choose \( s \) such that:
\[
2s^4 + (1 - 2s)^4 = 0.
\]
Solving this equation yields the appropriate value of \( s \), ensuring that:
\[
S_4(t) = e^{(A+B)t} + \mathcal{O}(t^5).
\]
Hence, \( S_4(t) \) is a fourth-order product formula.
We begin to observe a recurring pattern in the construction of higher-order product formulas. First, consider the identity:
\[
e^{(A+B)t} = e^{s(A+B)t} e^{(1 - 2s)(A+B)t} e^{s(A+B)t},
\]
which we used to build a third-order formula. This decomposition is, in fact, arbitrary. Any expression of the form
\[
e^{s_1(A+B)t} e^{s_2(A+B)t} \cdots e^{s_r(A+B)t},
\]
is valid as long as the coefficients satisfy the constraint
\[
\sum_{i=1}^r s_i = 1.
\]
This guarantees that the product of exponentials still approximates \( e^{(A+B)t} \), at least to first order.
The recursive structure observed in Suzuki's method suggests a general principle: to construct an order-\( k \) approximation from an order-\( (k-1) \) product formula, we define
\[
S_k(t) = S_{k-1}(s_1 t) S_{k-1}(s_2 t) \cdots S_{k-1}(s_r t),
\]
where \( \sum_{i=1}^r s_i = 1 \) to preserve the first-order behavior and in order to cancel the leading error term of order \( t^k \), it suffices to choose the coefficients \( s_i \) such that
\[
\sum_{i=1}^r s_i^k = 0.
\]
This recursive construction yields a product formula \( S_k(t) \) that approximates \( e^{(A+B)t} \) with an error of \( \mathcal{O}(t^{k+1}) \).
Let us generalize this result to the case of a sum of \( L \) non-commuting Hamiltonians, showing that Suzuki's recursive formula provides high-order approximations for arbitrary \( e^{t\sum_{j=1}^L H_j} \).

\begin{proposition}
    Let $H = \sum_{j=1}^L H_j \in M_{2^n}(\mathbb{C})$. Then
    \begin{equation*}
        e^{Ht} = e^{H_1 t/2} \cdots e^{H_{L-1} t/2} e^{H_L t} e^{H_{L-1} t/2} \cdots e^{H_1 t/2} + \mathcal{O}(t^3).
    \end{equation*}
\end{proposition}

\begin{proof}
    Let's use induction on the number of terms in the sum. For $L = 2$, we have previously shown that
    \[
    e^{(H_1 + H_2)t} = e^{H_1 t/2} e^{H_2 t} e^{H_1 t/2} + \mathcal{O}(t^3).
    \]
    Now assume that the statement holds for $L - 1$. Then,
    \begin{align*}
        e^{Ht} = e^{\left(H_1 + \sum_{j=2}^{L} H_j\right)t}
        &= e^{H_1 t/2} e^{\left(\sum_{j=2}^{L} H_j\right)t} e^{H_1 t/2} + \mathcal{O}(t^3) \\
        &= e^{H_1 t/2} \left( e^{H_2 t/2} \cdots e^{H_{L-1} t/2} e^{H_L t} e^{H_{L-1} t/2} \cdots e^{H_2 t/2} + \mathcal{O}(t^3) \right) e^{H_1 t/2} + \mathcal{O}(t^3) \\
        &= e^{H_1 t/2} \cdots e^{H_{L-1} t/2} e^{H_L t} e^{H_{L-1} t/2} \cdots e^{H_1 t/2} + \mathcal{O}(t^3).
    \end{align*}
\end{proof}

\begin{theorem} \label{suzuki}
    Let $H_j \in M_{2^n}(\mathbb{C})$, for $j \in \{1, \dots, L\}$ and let $S_{m-1}(t)$ be a product formula of order $(k-1)$ for $e^{t\sum_{j=1}^L H_j}$:
    \begin{equation*}
        e^{t\sum_{j=1}^LH_j} = S_{k-1}(t) + \mathcal{O}(t^k).
    \end{equation*}
    Then, the formula
    \begin{equation*}
        S_k(t)=\prod_{i=1}^r S_{k-1}(s_{i}t),
    \end{equation*}
    where $r \geq 3$, $\sum_{i=1}^r s_{i} = 1$ and $\sum_{i=1}^r s_{i}^k = 0$, is a product formula of order $k$:
    \begin{equation*}
        e^{t\sum_{j=1}^L H_j} = S_k(t) + \mathcal{O}(t^{k+1}).
    \end{equation*}
\end{theorem}
\begin{proof}
    We begin by splitting the full exponential into smaller segments, each scaled by a factor \( s_{i} \). This holds because the sum of all the \( s_{i} \) is 1.

    \begin{align*}
        e^{t\sum_{j=1}^L H_j} = \prod_{i=1}^r e^{s_{i}t\sum_{j=1}^L H_j}.
    \end{align*}
    
        Each \( e^{s_{i}t H} \) is approximated by the \((k-1)\)-order product formula up to error \( \mathcal{O}(t^k) \) and since the formula is of order \( k-1 \), the leading error term scales like \( s_{i}^k t^k \).

    \begin{align*}
        = \prod_{i=1}^r \left(S_{k-1}(s_{i}t) + s_{i}^k t^k R_k(\{H_j\}_j) + \mathcal{O}(t^{k+1})\right).
    \end{align*}

    We now multiply all these approximate formulas. The leading error at order \( t^k \) comes from summing the individual error terms, since cross-terms contribute only at higher order:
    \begin{align*}
        = \prod_{i=1}^r S_{k-1}(s_{i}t) + t^k R_k(\{H_j\}_j) \sum_{i=1}^r s_{i}^k + \mathcal{O}(t^{k+1}) = S_k(t) + \mathcal{O}(t^{k+1}).
    \end{align*}
\end{proof}
It should be noted that the system
\[
\sum_{i=1}^r s_i^k = 0 \quad \text{and} \quad \sum_{i=1}^r s_i = 1
\]
always admits a real solution when \( k \) is an odd integer and \( r \geq 3 \). To see this, consider setting
\[
s_1 = s_2 = \dots = s_{r-1} = a \in \mathbb{R}, \quad s_r = 1 - (r - 1)a.
\]
This ensures that \( \sum_{i=1}^r s_i = 1 \). We now solve
\[
(r - 1)a^k + (1 - (r - 1)a)^k = 0 \quad \Rightarrow \quad (r - 1)a^k = -\left(1 - (r - 1)a\right)^k.
\]
Taking the \( k \)-th root (notice that this step would not be valid if $k$ was even), we get:
\[
\sqrt[k]{r - 1}\, a = -1 + (r - 1)a \quad \Rightarrow \quad 1 = a((r - 1) - \sqrt[k]{r - 1}),
\]
which yields:
\[
a = \frac{1}{(r - 1) - \sqrt[k]{r - 1}}.
\]
In contrast, when \( k \) is even, such a solution cannot exist over the reals. Indeed, if \( k \) is even and all \( s_i \in \mathbb{R} \), then \( s_i^k \geq 0 \) and the sum \( \sum_{i=1}^r s_i^k = 0 \) implies that \( s_i = 0 \) for all \( i \). But then \( \sum_{i=1}^r s_i = 0 \), contradicting the normalization condition \( \sum_{i=1}^r s_i = 1 \). Therefore, complex coefficients are necessary in the even-\( k \) case.
\\\\
This is problematic if we want to avoid introducing complex time steps in the simulation. Fortunately, it turns out that by requiring our product formulas to satisfy one additional symmetry condition, we can systematically avoid the need for odd-order approximations altogether, thereby working exclusively with even-order formulas that admit real coefficients.

\begin{definition}
    We will say that a product formula $S(t) = \prod_{i=1}^{q} e^{p_{i,1} H_1 t} \cdots e^{p_{i,L} H_L t}$ is symmetric if $S(t)S(-t)=I$.
\end{definition}
\begin{theorem} \label{symmetric}
   Every symmetric product  formula  $S_{2k-1}(t)$  of odd order $2k-1$  is also of order $2k$.
\end{theorem}
\begin{proof}
    Since \( S_{2k-1}(t)S_{2k-1}(-t)=I \) and
    \[
    e^{t\sum_{j=1}^L H_j} = S_{2k-1}(t) + t^{2k}R_{2k}(\{H_j\}_j) + \mathcal{O}(t^{2k+1}),
    \]
    we have:
    \begin{align*}
        &\left( e^{t\sum_{j=1}^L H_j} + t^{2k}R_{2k}(\{H_j\}_j) + \mathcal{O}(t^{2k+1}) \right) 
          \left( e^{-t\sum_{j=1}^L H_j} + t^{2k}R_{2k}(\{H_j\}_j) + \mathcal{O}(t^{2k+1}) \right) = I.
    \end{align*}
    This can be rewritten as:
    \[
    \left( e^{t\sum_{j=1}^L H_j} + t^{2k}R_{2k}(\{H_j\}_j) \right) \left( e^{-t\sum_{j=1}^L H_j} + t^{2k}R_{2k}(\{H_j\}_j) \right) = I + \mathcal{O}(t^{2k+1}).
    \]
    Expanding the product, we obtain:
    \[
    t^{2k} R_{2k}(\{H_j\}_j) e^{t\sum_{j=1}^L H_j} + t^{2k} R_{2k}(\{H_j\}_j) e^{-t\sum_{j=1}^L H_j} = \mathcal{O}(t^{2k+1}).
    \]
    Dividing both sides by \( t^{2k} \), we find:
    \[
    R_{2k}(\{H_j\}_j) \left( e^{t\sum_{j=1}^L H_j} + e^{-t\sum_{j=1}^L H_j} \right) = \mathcal{O}(t).
    \]
    Substituting \( t = 0 \) yields:
    \[
    R_{2k}(\{H_j\}_j) \cdot 2I = 0 \quad \Rightarrow \quad R_{2k}(\{H_j\}_j) = 0.
    \]
\end{proof}
Thus, from the last two theorems, we obtain an infinite number of real symmetric decompositions of $e^{t\sum_{j=1}^LH_j}$ up to any order in $t$. \\\\ For example, recall that \( S_2(t) = e^{H_1t/2} \cdots e^{H_{L-1}t/2} e^{H_Lt} e^{H_{L-1}t/2} \cdots e^{H_1t/2} \). It is a simple exercise to verify that \( S_2(t) \) is symmetric. Now consider
\[
S_3(t) = S_2(st) S_2((1 - 2s)t) S_2(st),
\]
and apply Theorem~\ref{suzuki}. Since \( s + (1 - 2s) + s = 1 \) and \( 2s^3 + (1 - 2s)^3 = 0 \) for \( s = \frac{1}{2 - \sqrt[3]{2}} \), we obtain that \( S_3(t) \) is at least of order 3. However, notice that
\begin{align*}
S_3(t) S_3(-t) &= S_2(st) S_2((1 - 2s)t) S_2(st) S_2(-st) S_2(-(1 - 2s)t) S_2(-st) \\
&= S_2(st) S_2((1 - 2s)t) S_2(-(1 - 2s)t) S_2(-st) \\
&= S_2(st) S_2(-st) \\
&= I,
\end{align*}
so \( S_3(t) \) is symmetric. By Theorem~\ref{symmetric}, it is therefore of order at least 4. In other words, we can define \( S_4(t) = S_3(t) \). Repeating this process with
\[
S_5(t) = S_4(st) S_4((1 - 2s)t) S_4(st),
\]
for \( s = \frac{1}{2 - \sqrt[5]{2}} \), we find that it is of at least order 6 and we can write \( S_6(t) = S_5(t) \). In general, we have found a recursive even-order product formula given by
\begin{equation*}
    S_{2k}(t) = S_{2k-2}(s_k t) S_{2k-2}((1 - 2s_k)t) S_{2k-2}(s_k t),
\end{equation*}
where \( s_k = \frac{1}{2 - \sqrt[2k-1]{2}} \) and the base case is
\begin{equation*}
    S_2(t) = e^{H_1t/2} \cdots e^{H_{L-1}t/2} e^{H_Lt} e^{H_{L-1}t/2} \cdots e^{H_1t/2}.
\end{equation*}
This construction yields a symmetric product formula of order \( 2k \), meaning
\begin{equation*}
    e^{Ht} = S_{2k}(t) + \mathcal{O}(t^{2k+1}).
\end{equation*}
Unfortunately, we cannot use this specific recursion formula in practice. Notice that \( s_k > 1 \) for all \( k \). Consequently, this series of approximants is not convergent in the limit \( k \to \infty \). Thus, this scheme of decomposition is not practical for large \( k \). 
Let's try to illustrate this (a more detailed discussion can be found in \cite{suzuki1994convergence}). We can see the exponential map \( e^{x(A+B)} \) as a time-evolution map from the time \( t = 0 \) to the time \( t = x \). In the formula for \( S_4(x) \), the term \( S_2(s_2x) \) on the right approximates the time evolution from \( t = 0 \) to \( t = s_2x \approx 1.35x \), the term \( S_2((1 - 2s_2)x) \) in the middle approximates the time evolution from \( t = sx \) to \( t = s_2x + (1 - 2s_2)x = (1 - s_2)x \approx -0.35x \) and the term \( S_2(sx) \) on the left approximates the time evolution from \( t = (1 - s_2)x \) to \( t = (1 - s_2)x + s_2x = x \).
Let us express this time evolution in Figure \ref{fig:time-evolution-diagram}:
 \begin{figure}[H]
    \centering
\begin{tikzpicture}[scale=0.5]
\draw[thick,dashed] (-3,3) -- (22,3);
\draw[thick,dashed] (-3,8) -- (22,8);
\draw[thick] (-3,10) -- (-3,0);
\draw[thick,-latex] (-3,3) -- (-3,11);
\draw[thick,-latex] (-2,3) -- (-2,6.5);
\draw[thick] (-2,6) -- (-2,10);
\draw[thick] (-2,10) -- (-1.5,10);
\draw[thick, -latex] (-1.5,10) -- (-1.5,5);
\draw[thick] (-1.5,6) -- (-1.5,2);
\draw[thick] (-1.5,2) -- (-1,2);
\draw[thick, -latex] (-1,2) -- (-1,5);
\draw[thick] (-1,4) -- (-1,8);

\draw[thick,-latex] (1,3) -- (1,6.5);
\draw[thick] (1,6) -- (1,12);
\draw[thick] (1,12) -- (1.5,12);
\draw[thick, -latex] (1.5,12) -- (1.5,5);
\draw[thick] (1.5,6) -- (1.5,2);
\draw[thick] (1.5,2) -- (2,2);
\draw[thick, -latex] (2,2) -- (2,5);
\draw[thick] (2,5) -- (2,10);
\draw[thick] (2,10) -- (2.5,10);

\draw[thick, -latex] (2.5,10) -- (2.5,6);
\draw[thick] (2.5,6) -- (2.5,1);
\draw[thick] (2.5,1) -- (3,1);
\draw[thick, -latex] (3,1) -- (3,5.5);
\draw[thick] (3,4) -- (3,12);
\draw[thick] (3,12) -- (3.5,12);
\draw[thick, -latex] (3.5,12) -- (3.5,4.5);
\draw[thick] (3.5,6) -- (3.5,2);

\draw[thick] (3.5,2) -- (4,2);
\draw[thick,-latex] (4,2) -- (4,6.5);
\draw[thick] (4,6) -- (4,10);
\draw[thick] (4,10) -- (4.5,10);
\draw[thick, -latex] (4.5,10) -- (4.5,5);
\draw[thick] (4.5,6) -- (4.5,1);
\draw[thick] (4.5,1) -- (5,1);
\draw[thick, -latex] (5,1) -- (5,5);
\draw[thick] (5,4) -- (5,8);

\draw[thick,-latex] (7,3) -- (7,6.5);
\draw[thick] (7,6) -- (7,14);
\draw[thick] (7,14) -- (7.5,14);
\draw[thick, -latex] (7.5,14) -- (7.5,5);
\draw[thick] (7.5,6) -- (7.5,2);
\draw[thick] (7.5,2) -- (8,2);
\draw[thick, -latex] (8,2) -- (8,5);
\draw[thick] (8,5) -- (8,12);
\draw[thick] (8,12) -- (8.5,12);

\draw[thick, -latex] (8.5,12) -- (8.5,6);
\draw[thick] (8.5,6) -- (8.5,1);
\draw[thick] (8.5,1) -- (9,1);
\draw[thick, -latex] (9,1) -- (9,5.5);
\draw[thick] (9,4) -- (9,14);
\draw[thick] (9,14) -- (9.5,14);
\draw[thick, -latex] (9.5,14) -- (9.5,4.5);
\draw[thick] (9.5,6) -- (9.5,2);

\draw[thick] (9.5,2) -- (10,2);
\draw[thick,-latex] (10,2) -- (10,6.5);
\draw[thick] (10,6) -- (10,12);
\draw[thick] (10,12) -- (10.5,12);
\draw[thick, -latex] (10.5,12) -- (10.5,5);
\draw[thick] (10.5,6) -- (10.5,1);
\draw[thick] (10.5,1) -- (11,1);
\draw[thick, -latex] (11,1) -- (11,5);
\draw[thick] (11,4) -- (11,10);

\draw[thick] (11,10) -- (11.5,10);
\draw[thick, -latex] (11.5,10) -- (11.5,6);
\draw[thick] (11.5,6) -- (11.5,0);
\draw[thick] (11.5,0) -- (12,0);
\draw[thick, -latex] (12,0) -- (12,5.5);
\draw[thick] (12,4) -- (12,12);
\draw[thick] (12,12) -- (12.5,12);
\draw[thick, -latex] (12.5,12) -- (12.5,4.5);
\draw[thick] (12.5,6) -- (12.5,1);

\draw[thick] (12.5,1) -- (13,1);
\draw[thick,-latex] (13,1) -- (13,6.5);
\draw[thick] (13,6) -- (13,14);
\draw[thick] (13,14) -- (13.5,14);
\draw[thick, -latex] (13.5,14) -- (13.5,5);
\draw[thick] (13.5,6) -- (13.5,0);
\draw[thick] (13.5,0) -- (14,0);
\draw[thick, -latex] (14,0) -- (14,5);
\draw[thick] (14,4) -- (14,12);

\draw[thick] (14,12) -- (14.5,12);
\draw[thick, -latex] (14.5,12) -- (14.5,6);
\draw[thick] (14.5,6) -- (14.5,1);
\draw[thick] (14.5,1) -- (15,1);
\draw[thick, -latex] (15,1) -- (15,5.5);
\draw[thick] (15,4) -- (15,14);
\draw[thick] (15,14) -- (15.5,14);
\draw[thick, -latex] (15.5,14) -- (15.5,4.5);
\draw[thick] (15.5,6) -- (15.5,2);

\draw[thick] (15.5,2) -- (16,2);
\draw[thick,-latex] (16,2) -- (16,6.5);
\draw[thick] (16,6) -- (16,12);
\draw[thick] (16,12) -- (16.5,12);
\draw[thick, -latex] (16.5,12) -- (16.5,5);
\draw[thick] (16.5,6) -- (16.5,1);
\draw[thick] (16.5,1) -- (17,1);
\draw[thick, -latex] (17,1) -- (17,5);
\draw[thick] (17,5) -- (17,10);
\draw[thick] (17,10) -- (17.5,10);

\draw[thick, -latex] (17.5,10) -- (17.5,6);
\draw[thick] (17.5,6) -- (17.5,0);
\draw[thick] (17.5,0) -- (18,0);
\draw[thick, -latex] (18,0) -- (18,5.5);
\draw[thick] (18,4) -- (18,12);
\draw[thick] (18,12) -- (18.5,12);
\draw[thick, -latex] (18.5,12) -- (18.5,4.5);
\draw[thick] (18.5,6) -- (18.5,1);

\draw[thick] (18.5,1) -- (19,1);
\draw[thick,-latex] (19,1) -- (19,6.5);
\draw[thick] (19,6) -- (19,10);
\draw[thick] (19,10) -- (19.5,10);
\draw[thick, -latex] (19.5,10) -- (19.5,5);
\draw[thick] (19.5,6) -- (19.5,0);
\draw[thick] (19.5,0) -- (20,0);
\draw[thick, -latex] (20,0) -- (20,5);
\draw[thick] (20,4) -- (20,8);

\node[left] at (-3,3) {$0$};
\node[left] at (-3,8) {$x$};
\node[left] at (-3,10) {$t$};
\node at (-1,1) {(a)};
\node at (3.75,1) {(b)};
\node at (15.75,1) {(c)};
\end{tikzpicture}
\caption{Diagrams that represent the time evolution of 
    (a) the fourth-order product formula \(S_4(x)\), 
    (b) the sixth-order product formula \(S_6(x)\) and 
    (c) the eighth-order product formula \(S_8(x)\) in terms of $S_2(x)$. 
    Notice that in (b) and especially (c), the evolution steps move increasingly further backward and forward in time}
    \label{fig:time-evolution-diagram}
\end{figure}
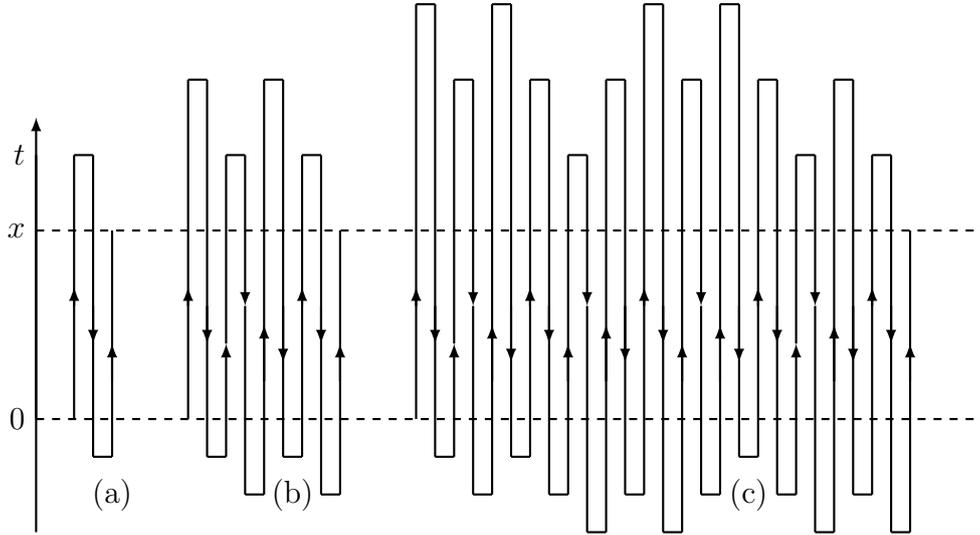
As is evident, the evolution according to $S_4(x)$ has a part that goes into the ``past''  or $t < 0$. This can be problematic in systems where negative time is either undefined or has no physical meaning. Moreover, as the recursion level increases, the approximation ``explodes'' and becomes unstable.\\\\ Let's consider an alternative recursion formula:
\[
S_{2k}(t) = \left(S_{2k-2}(s_k t)\right)^2 S_{2k-2}((1-4s_k)t) \left(S_{2k-2}(s_k t)\right)^2.
\]
Notice that starting with \( S_2(t) \), which is symmetric, the following product formulas generated by this recursion will also be symmetric. Therefore, we can write it only for even orders. By applying Theorem \ref{suzuki}, it suffices to choose \( s_k \) such that:
\[
4s_k^{2k-1} + (1 - 4s_k)^{2k-1} = 0.
\]
This leads to the solution:
\[
s_k = \frac{1}{4 - \sqrt[2k-1]{4}},
\]
and \( s_k < 1 \) for all \( k \). Now, let's visualize the time evolution using this new recursive formula in Figure \ref{fig:recursion_formula}.

\begin{figure}[H]
    \centering
\includegraphics[width=\textwidth]{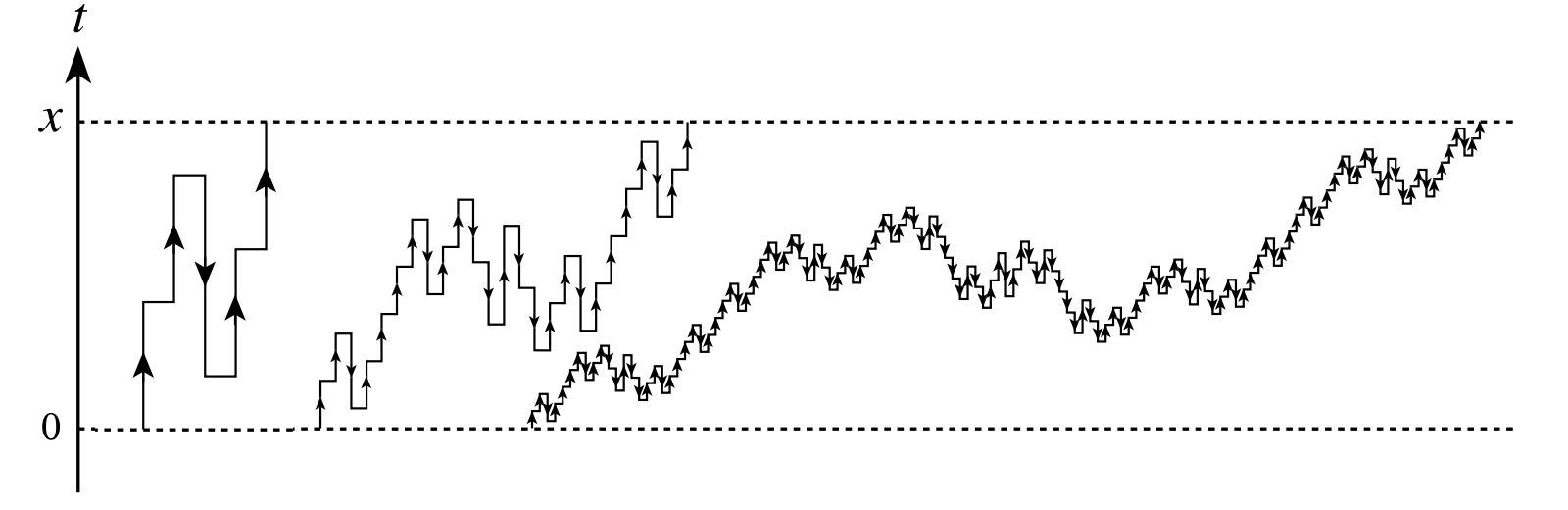}
\caption{Image from \cite{Hatano2005} showing the fourth-, sixth- and eighth-order product formulas using the new recursion scheme}
\label{fig:recursion_formula}
\end{figure}
We can continue this recursive procedure, ultimately arriving at the exact time evolution, where the diagram looks like a fractal. This approach is often referred to as fractal decomposition. Interestingly, the back-and-forth time evolution in a fractal manner accurately reproduces the exact time evolution.
 \\\\ Summing up, we have arrived at
\[
S_{2k}(t) = \left(S_{2k-2}(s_k t)\right)^2 S_{2k-2}((1-4s_k)t) \left(S_{2k-2}(s_k t)\right)^2,
\]
with
\[
s_k = \frac{1}{4 - \sqrt[2k-1]{4}},
\]
and
\[
S_2(t) = e^{H_1 t / 2} \cdots e^{H_{L-1} t / 2} e^{H_L t} e^{H_{L-1} t / 2} \cdots e^{H_1 t / 2}.
\]
Thus, we have
\[
e^{Ht} = S_{2k}(t) + \mathcal{O}(t^{2k+1}),
\]
and therefore,
\[
e^{-iHt} = S_{2k}(-it) + \mathcal{O}(t^{2k+1}).
\]
\\\\In order to calculate the number of gates required to simulate a Hamiltonian using this new product formula, we need to estimate the number of \( e^{-ip_{i,j}H_jt} \) terms that appear in it. From this, we can infer that, using the Solovay-Kitaev theorem, the total number of gates will have only a polylogarithmic overhead. Notice that, for the first time, it is not immediately clear how many \( e^{-ip_{i,j}H_jt} \) terms we are working with since the product formula is defined recursively. In the following result, we provide an upper bound based on the one from \cite{berry2007efficient}.

\begin{theorem} \label{nexp}
Let \( N_{\text{exp}} \) denote the number of exponentials in 
\[
S_{2k}(t) = \left(S_{2k-2}(s_kt)\right)^2 S_{2k-2}((1-4s_k)t) \left(S_{2k-2}(s_kt)\right)^2,
\]
with \( s_k = \frac{1}{4 - \sqrt[2k-1]{4}} \) and $S_2(t) = \prod_{j=1}^{L} e^{H_jt / 2} \prod_{j'=L}^{1} e^{H_{j'} t / 2}$. When the permissible error is bounded by \( \epsilon \), \( N_{\text{exp}} \) is bounded by
\[
N_{\text{exp}} \leq \frac{L 5^{2k} (L \tau)^{1 + 1/2k}}{\epsilon^{1/2k}},
\] for \( \epsilon \leq 1 \leq 2L 5^{k-1} \tau \), where \( \tau = \max_{j} \| H_j \|
 t \) and \( k \) is an arbitrary positive integer.
\end{theorem}
First we need the following lemma:
\begin{lemma}
Let \( m \in \mathbb{N} \) satisfy
\[
\frac{2 L 5^{k-1} \tau}{m} \leq 1.
\]
Then, the following bound holds:
\[
\left\| e^{-i t \sum_{j=1}^L H_j} - [S_{2k} (-it / m)]^m \right\| \leq \frac{ (2 L 5^{k-1} \tau)^{2k+1}}{m^{2k}}.
\]
\end{lemma}

\begin{proof}
We use the fact that the symmetric product formula \( S_{2k}(t) \) approximates the exponential \( e^{t \sum_{j=1}^L H_j} \) to order \( 2k+1 \), i.e.,
\[
e^{t \sum_{j=1}^L H_j} = S_{2k}(t) + \mathcal{O}(t^{2k+1}).
\]
This means the Taylor expansions of both sides agree up to order \( 2k \), so the difference is bounded by the tails starting at \( t^{2k+1} \). The exponential has Taylor expansion:
\[
e^{t \sum_{j=1}^L H_j} = \sum_{l=0}^\infty \frac{t^l}{l!} \left( \sum_{j=1}^L H_j \right)^l.
\]
Since the \( H_j \) are not necessarily commuting, \( \left( \sum_j H_j \right)^l \) can have up to \( L^l \) terms. Defining \( \Lambda := \max_j \| H_j \| \), we have:

\[
\left\| \frac{t^l}{l!} \left( \sum_{j=1}^{L} H_j \right)^l \right\| \leq \frac{ |t \Lambda|^l L^l }{l! }.
\]
Next, consider \( S_{2k}(t) \), which consists of
\[
N := 2(L-1)5^{k-1} + 1
\]
exponentials (this is straightforward to prove by induction starting with $S_2(t))$ of the form \( e^{\tilde{s}_i t H_{j_i}} \), where each coefficient \( \tilde{s}_i \in \{s_k, 1 - 4s_k\} \) or is a product of other $\tilde{s}_i$. Therefore, \( |\tilde{s}_i| < 1 \) for all $i$. Each exponential can be expanded as:
\[
e^{\tilde{s}_i t H_{j_i}} = \sum_{n_i=0}^\infty \frac{(\tilde{s}_i t H_{j_i})^{n_i}}{n_i!}.
\]
The full product \( S_{2k}(t) \) is then:
\[
S_{2k}(t) = \prod_{i=1}^N \left( \sum_{n_i=0}^{\infty} \frac{(\tilde{s}_i t H_{j_i})^{n_i}}{n_i!} \right),
\]
which produces terms of the form
\[
t^l \cdot \frac{\tilde{s}_1^{n_1} \cdots \tilde{s}_N^{n_N}}{n_1! \cdots n_N!} \cdot H_{j_1}^{n_1} \cdots H_{j_N}^{n_N}, \quad \text{where } \sum n_i = l.
\]
The norm of these terms is bounded by:
 \[ \left\|\sum_{\sum n_i = l}t^l \cdot \frac{\tilde{s}_1^{n_1} \cdots \tilde{s}_N^{n_N}}{n_1! \cdots n_N!} \cdot H_{j_1}^{n_1} \cdots H_{j_N}^{n_N} \right\|  \leq 
|t|^l \cdot \Lambda^l \sum_{\sum n_i = l} \frac{1}{n_1! \cdots n_N!} = \frac{|t|^l \cdot \Lambda^l N^l}{l!},
\]
by the multinomial theorem.
Combining the bounds for the exponential and for \( S_{2k}(t) \), we get:
\begin{align*}
\left\| e^{t \sum_{j=1}^L H_j} - S_{2k}(t) \right\|
&\leq \sum_{l=2k+1}^{\infty} \left(\frac{ t^l\Lambda^{l} L^{l}}{l!} + \frac{t^l\Lambda^{l}[2(L-1)5^{k-1}+1]^{l}}{l!} \right) \\  & = \sum_{l=2k+1}^{\infty}\frac{|t\Lambda|^{l}}{l!}\big{\{}L^{l}+[2(L-1)5^{k-1}+1]^{l}\big{\}} \\
&\leq 2\sum_{l=2k+1}^{\infty}\frac{|t\Lambda|^{l}}{l!}[2L5^{k-1}]^{l} \\
&= \frac{2}{(2k+1)!}[2L5^{k-1}t\Lambda]^{2k+1}\sum_{l=2k+1}^{\infty}\frac{|2L5^{k-1}t\Lambda|^{l-(2k+1)}(2k+1)!}{l!} \\
&= \frac{2}{(2k+1)!}[2L5^{k-1}t\Lambda]^{2k+1}\sum_{m=0}^{\infty}\frac{|2L5^{k-1}t\Lambda|^{m}(2k+1)!}{(2k+1+m)!} \\
& \leq \frac{2}{(2k+1)!}[2L5^{k-1}t\Lambda]^{2k+1}\sum_{m=0}^{\infty}\frac{|2L5^{k-1}t\Lambda|^{m}}{m!} \\
&\leq (1/3)[2L5^{k-1}t\Lambda]^{2k+1}e^{[2L5^{k-1}t\Lambda]}.
\end{align*}
Therefore we obtain the inequality
\[
\Big{\|}e^{t\sum_{j=1}^{L}H_{j}}-S_{2k}(t)\Big{\|} \leq [2L5^{k-1}\Lambda t]^{2k+1},
\]
provided $|2L5^{k-1}\Lambda t|\leq 1$. Substituting $t$ by $-it/m$  and applying Lemma \ref{linear_error}, gives the error bound
\[
\Big{\|}e^{-it\sum_{j=1}^{L}H_{j}}-[S_{2k}(-it/m)]^{m}\Big{\|} =
\Big{\|}(e^{-it/m\sum_{j=1}^{L}H_{j}})^m-[S_{2k}(-it/m)]^{m}\Big{\|} 
 \leq \frac{(2 L 5^{k-1} \tau)^{2k+1}}{m^{2k}}, 
\] for $2L5^{k-1}\Lambda t/m\leq 1$. 
\end{proof}
\begin{definition}
Let $x \in \mathbb{R}$. The integer part (or floor function) of $x$, denoted by $\lfloor x \rfloor$, is defined as the greatest integer less than or equal to $x$. Formally, it is given by
\[
\lfloor x \rfloor = \max \left\{ n \in \mathbb{Z} \mid n \leq x \right\}.
\]
Equivalently, the integer part of $x$ is the unique integer $n \in \mathbb{Z}$ such that
\[
n \leq x < n + 1.
\]
\end{definition}
\begin{definition}
Let $x \in \mathbb{R}$. The ceiling function, denoted by $\lceil x \rceil$, is defined as the smallest integer greater than or equal to $x$. This can also be expressed as the integer part of $x$ plus one, i.e.,
\[
\lceil x \rceil = \lfloor x \rfloor + 1.
\]
\end{definition}
Now we can prove the theorem
\begin{proof}[Proof of Theorem \ref{nexp}]
Let us take
\[
m = \left\lceil(2L5^{k-1}\tau)^{1 + \frac{1}{2k}}/\epsilon^{1/2k} \right\rceil.
\]
Then $(2L5^{k-1}\Lambda t/m)^{2k}\leq(2L5^{k-1}\Lambda t)^{2k+1}/m^{2k} \leq  \epsilon \leq 1$ and therefore $2L5^{k-1}\Lambda t/m \leq 1$. We can thus apply the previous lemma. \[
\Big{\|}e^{-it\sum_{j=1}^{L}H_{j}}-[S_{2k}(-it/m)]^{m}\Big{\|} \leq \frac{(2 L 5^{k-1} \tau)^{2k+1}}{m^{2k}} \leq \epsilon.
\]
Because the number of exponentials in $S_{2k}(-it/m)$ does not exceed $2L5^{k-1}$, we have
\[
N_{\exp} \leq 2L5^{k-1}m.
\]
Using that $\lceil x \rceil \leq x+1 $ for all $x \in \mathbb{R}$:
\begin{align*}
N_{\exp} &\leq 2L5^{k-1}(2L5^{k-1}\tau)^{1 + \frac{1}{2k}}/\epsilon^{1/2k} + 2L5^{k-1} \\ &\leq 2^{2 + \frac{1}{2k}}L5^{2 k - 1/2 k - 3/2}(L\tau)^{1 +\frac{1}{2k}}/\epsilon^{1/2k}+ 2L5^{k-1} \\ &\leq \frac{2^{2+\frac{1}{2k}}}{5^{\frac{3}{2}+\frac{1}{2k}}}L5^{2 k}(L\tau)^{1 +\frac{1}{2k}}/\epsilon^{1/2k} + 2L5^{k-1} \\  &\leq \frac{4\sqrt{5}}{25}L5^{2 k}(L\tau)^{1 +\frac{1}{2k}}/\epsilon^{1/2k}+ 2L5^{k-1} \\ &\leq L5^{2 k}(L\tau)^{1 +\frac{1}{2k}}/\epsilon^{1/2k}.
\end{align*}
This concludes the proof, provided that:\[
\frac{4\sqrt{5}}{25} \cdot \frac{L5^{2k}(L\tau)^{1 +\frac{1}{2k}}}{\epsilon^{1/2k}} + 2L5^{k-1} 
\leq 
\frac{L5^{2k}(L\tau)^{1 +\frac{1}{2k}}}{\epsilon^{1/2k}}.
\]
To verify this, divide both sides by the positive quantity \( \frac{L5^{2k}(L\tau)^{1 + \frac{1}{2k}}}{\epsilon^{1/2k}} \), preserving the inequality direction:
\[
\frac{4\sqrt{5}}{25} + \frac{2L5^{k-1} \cdot \epsilon^{1/2k}}{L5^{2k}(L\tau)^{1 + \frac{1}{2k}}} \leq 1.
\]
Now simplify the second term:
\[
\frac{2L5^{k-1} \cdot \epsilon^{1/2k}}{L5^{2k}(L\tau)^{1 + \frac{1}{2k}}}
= \frac{2 \cdot \epsilon^{1/2k}}{5^{k+1}(L\tau)^{1 + \frac{1}{2k}}}.
\]
Thus, we must show:
\[
\frac{4\sqrt{5}}{25} + \frac{2 \cdot \epsilon^{1/2k}}{5^{k+1}(L\tau)^{1 + \frac{1}{2k}}} \leq 1.
\]
From the condition \( 1 \leq 2L5^{k-1}\tau \), we can deduce:
\[
L\tau \geq \frac{1}{2 \cdot 5^{k-1}} = \frac{5}{2 \cdot 5^k}.
\]
Raising both sides to the power \( 1 + \frac{1}{2k} \), we get:
\[
(L\tau)^{1 + \frac{1}{2k}} 
\geq 
\left( \frac{5}{2 \cdot 5^k} \right)^{1 + \frac{1}{2k}} 
= \frac{5^{1 + \frac{1}{2k}}}{2^{1 + \frac{1}{2k}} \cdot 5^{k + \frac{1}{2}}}.
\]
Substitute this lower bound into the second term:
\[
\frac{2 \cdot \epsilon^{1/2k}}{5^{k+1}(L\tau)^{1 + \frac{1}{2k}}}
\leq
\frac{2 \cdot \epsilon^{1/2k} \cdot 2^{1 + \frac{1}{2k}} \cdot 5^{k + \frac{1}{2}}}
{5^{k+1} \cdot 5^{1 + \frac{1}{2k}}}
= 2^{2 + \frac{1}{2k}} \cdot 5^{ -\frac{3}{2} - \frac{1}{2k} } \cdot \epsilon^{1/2k}.
\]
Now, observe that \( \epsilon^{1/2k} \leq 1 \) for \( \epsilon \leq 1 \) and the constants satisfy:
\[
2^{2 + \frac{1}{2k}} \cdot 5^{ -\frac{3}{2} - \frac{1}{2k} } \leq \frac{4\sqrt{5}}{25},
\]
for all \( k \geq 1 \). Therefore:
\[
\frac{2 \cdot \epsilon^{1/2k}}{5^{k+1}(L\tau)^{1 + \frac{1}{2k}}}
\leq
\frac{4\sqrt{5}}{25}.
\]
Adding both terms:
\[
\frac{4\sqrt{5}}{25} + \frac{4\sqrt{5}}{25} = \frac{8\sqrt{5}}{25} \leq 1,
\]
\end{proof}
\begin{remark}
    As we will see in the next section, the bound we have obtained for \( N_{\text{exp}} \) is actually quite loose. Therefore, the results and bounds given for the rest of this section should be taken as educational and illustrative of what kinds of analyses one can perform given such a bound, rather than as state-of-the-art results.
\end{remark}
Let us consider the bound we have just obtained:
\[
N_{\text{exp}} \leq \frac{L \cdot 5^{2k} \cdot (L \tau)^{1 + \frac{1}{2k}}}{\epsilon^{\frac{1}{2k}}}.
\]
Notice that by taking \( k \) sufficiently large, it is possible to achieve a scaling that is arbitrarily close to linear in \( \tau \). However, for fixed \( \tau \) and \( \epsilon \), increasing \( k \) eventually increases \( N_{\text{exp}} \). 
We therefore aim to estimate the optimal value of \( k \) by minimizing the function:
\[
f(k) = \frac{L \cdot 5^{2k} \cdot (L \tau)^{1 + \frac{1}{2k}}}{\epsilon^{\frac{1}{2k}}},
\]
where \( L, \tau, \epsilon > 0 \) and \( k >0 \).
Taking the natural logarithm of \( f(k) \), we get:
\[
\ln f(k) = \ln L + 2k \ln 5 + \left(1 + \frac{1}{2k}\right) \ln(L \tau) - \frac{1}{2k} \ln \epsilon.
\]
Letting
\[
C := \ln\left( \frac{L \tau}{\epsilon} \right),
\]
we simplify:
\[
\ln f(k) = \ln (L^2 \tau) + 2k \ln 5 + \frac{C}{2k}.
\]
To find the extremum, we compute the derivative with respect to \( k \):
\[
\frac{d}{dk} \ln f(k) = 2 \ln 5 - \frac{C}{2k^2}.
\]
Setting this to zero yields the critical point:
\[
2 \ln 5 = \frac{C}{2k^2} \quad \Rightarrow \quad k^* = \frac{1}{2} \sqrt{ \frac{C}{\ln 5} }.
\]
Substituting back for \( C \), we get:
\[
k^* = \frac{1}{2} \sqrt{ \log_5 \left( \frac{L \tau}{\epsilon} \right) }.
\]
To verify this is indeed a minimum, observe that the second derivative is
\[
\frac{d^2}{dk^2} \ln f(k) = \frac{C}{k^3} > 0 \quad \text{for all } k > 0,
\]
which confirms that \( k^* \) minimizes \( f(k) \). Thus, \( f(k) \) is strictly decreasing on \( (0, k^*) \) and increasing on \( (k^*, \infty) \).
Since \( k \) must be a positive integer, a natural heuristic choice is:
\[
k = \text{round} \left[ \frac{1}{2} \sqrt{ \log_5 \left( \frac{L \tau}{\epsilon} \right) + 1 } \right].
\]
 We now use this to obtain a bound for \( N_{\text{exp}} \) that is independent of \( k \). Recall that
\[
C = \log_5\left( \frac{L \tau}{\epsilon} \right), \quad \text{so} \quad \epsilon = \frac{L \tau}{5^C}.
\]
Substituting into the original expression gives:
\begin{align*}
N_{\text{exp}} 
&\leq L \cdot 5^{2k} \cdot (L \tau)^{1 + \frac{1}{2k}} \cdot \left( \frac{L \tau}{5^C} \right)^{-\frac{1}{2k}} \\
&= L^2 \tau \cdot 5^{2k + \frac{C}{2k}}.
\end{align*}
Since \( k \) is chosen such that
\[
\left| k - \frac{1}{2} \sqrt{C + 1} \right| \leq \frac{1}{2},
\]
we have the bounds
\[
\frac{1}{2} \sqrt{C + 1} - \frac{1}{2} \leq k \leq \frac{1}{2} \sqrt{C + 1} + \frac{1}{2},
\]
and hence,
\[
2k \leq \sqrt{C + 1} + 1.
\]
Using the lower bound on \( k \),
\[
\frac{C}{2k} \leq \frac{C}{\sqrt{C + 1} - 1}.
\]
Rationalizing the denominator:
\[
\frac{C}{\sqrt{C + 1} - 1} = \frac{C(\sqrt{C + 1} + 1)}{(\sqrt{C + 1} - 1)(\sqrt{C + 1} + 1)} = \sqrt{C + 1} + 1.
\]
Adding both components, we obtain:
\[
2k + \frac{C}{2k} \leq 2\sqrt{C + 1} + 2.
\]
Substituting this into our earlier expression for \( N_{\text{exp}} \), we get:
\[
N_{\text{exp}} \leq L^2 \tau \cdot 5^{2\sqrt{C + 1} + 2} = 25 \cdot L^2 \tau \cdot 5^{2\sqrt{C + 1}}.
\]
Finally, substituting back the expression for \( C \), we arrive at a bound for $N_{\text{exp}}$ which is independent of $k$:
\[
N_{\text{exp}} \leq 25 \cdot L^2 \tau \cdot 5^{2\sqrt{ \log_5\left( \frac{L \tau}{\epsilon} \right) + 1 }}.
\]
\subsection{Application to the Ising Hamiltonian}
Let us apply the new bound to the 1D Ising Hamiltonian for \( n \) particles:
\[
H = \sum_{i=1}^{n-1} Z_i Z_{i+1} + \sum_{i=1}^n X_i.
\]
First, we determine the number of terms \( L \) in the Hamiltonian. There are \( n - 1 \) two-qubit interaction terms \( Z_i Z_{i+1} \) and \( n \) single-qubit field terms \( X_i \), giving:
\[
L = (n - 1) + n = 2n - 1.
\]
Next, we evaluate \( \tau = \max_j \|H_j\| \cdot t \), where \( H_j \) denotes an individual term in the Hamiltonian and \( t \) is the total simulation time. Each \( H_j \) is a tensor product of Pauli and identity matrices, all of which are Hermitian with spectral norm 1 (i.e., their largest eigenvalue in absolute value is 1). Thus, \( \|H_j\| = 1 \) for all \( j \). 
We take \( t = n \), a common choice in the literature, since the system must evolve for a time proportional to \( n \) for information to propagate across the entire chain, in accordance with the Lieb-Robinson bound \cite{lieb1972finite}. Therefore,
\[
\tau = 1 \cdot n = n.
\]
We set the simulation error tolerance to \( \epsilon = 10^{-3} \). Substituting into the expression for the number of exponentials required for the simulation, we obtain:
\[
N_{\text{exp}} \leq 25 \cdot (2n - 1)^2 \cdot n \cdot 5^{2\sqrt{ \log_5\left( \frac{(2n - 1) \cdot n}{10^{-3}} \right) + 1 }}.
\]
This simplifies to:
\[
N_{\text{exp}} \leq 25 \cdot (2n - 1)^2 \cdot n \cdot 5^{2\sqrt{ \log_5\left( 10^3 \cdot n(2n - 1) \right) + 1 }}.
\]
Using logarithmic identities:
\[
\log_5\left( 10^3 \cdot n(2n - 1) \right) = \log_5(10^3) + \log_5(n) + \log_5(2n - 1),
\]
and noting that \( \log_5(10^3) = \frac{3}{\log_{10}(5)} \), we can write the exponent explicitly as:
\[
2\sqrt{ \frac{3}{\log_{10}(5)} + \log_5(n) + \log_5(2n - 1) + 1 }.
\]
Thus, the exact symbolic expression for the number of exponentials required to simulate the time evolution of the 1D Ising model on \( n \) qubits for time \( t = n \) with error at most \( \epsilon = 10^{-3} \) is:
\[
N_{\text{exp}} \leq 25 \cdot (2n - 1)^2 \cdot n \cdot 5^{2\sqrt{ \frac{3}{\log_{10}(5)} + \log_5(n) + \log_5(2n - 1) + 1 }}.
\]
Finally, we plot the bound on the number of exponentials against the system size on Figure \ref{fig:ising}:

\begin{figure}[H]
    \centering
    \includegraphics[width=\textwidth]{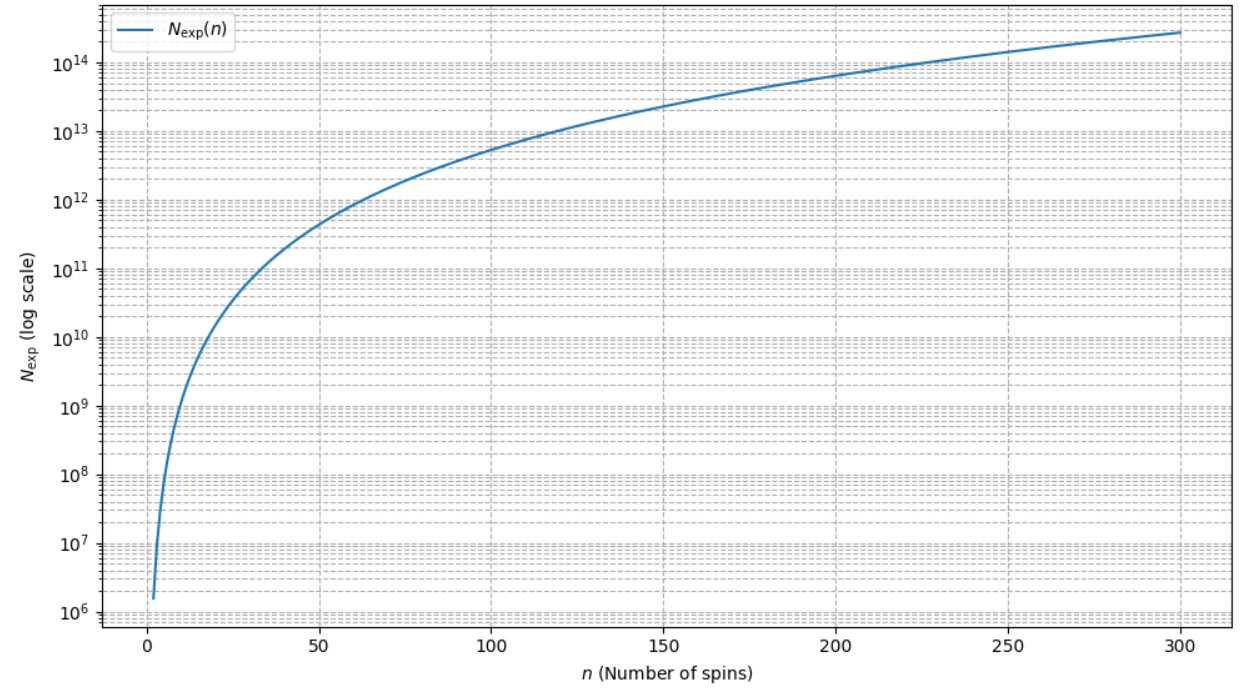}
    \caption{Scaling of \( N_{\text{exp}} \) bound for 1D Ising Model Simulation}
    \label{fig:ising}
\end{figure}

\section{Numerical Validation of the Properties of Suzuki's Product Formulas}
For the purpose of numerical validation, we consider the three matrices
\[
A = \begin{bmatrix}
2.20 & 6.90 \\
4.20 & 6.66
\end{bmatrix}, \quad
B = \begin{bmatrix}
1.10 & 6.90 \\
0 & 3.33
\end{bmatrix}, \quad
C = \begin{bmatrix}
1.10 & 0 \\
4.20 & 3.33
\end{bmatrix},
\]
which satisfy \( A = B + C \), but \( BC \ne CB \).
\\\\
The first experiment confirms the theoretical scaling behavior of the error with respect to the evolution time \( t \). For each product formula order \( k \in \{2, 4, 6\} \), we apply a single-step approximation \( S_{2k}(t) \) to approximate \( e^{At} \). The approximation is evaluated at 10 values of \( t \) logarithmically spaced in the interval \( [10^{-2}, 10^{-1}] \), meaning the logarithms of the values are evenly spaced. The relative spectral norm error is then plotted against \( t \) on a logarithmic scale in Figure \ref{fig:error-vs-time}.
\begin{figure}[H]
    \centering
    \includegraphics[width=\textwidth]{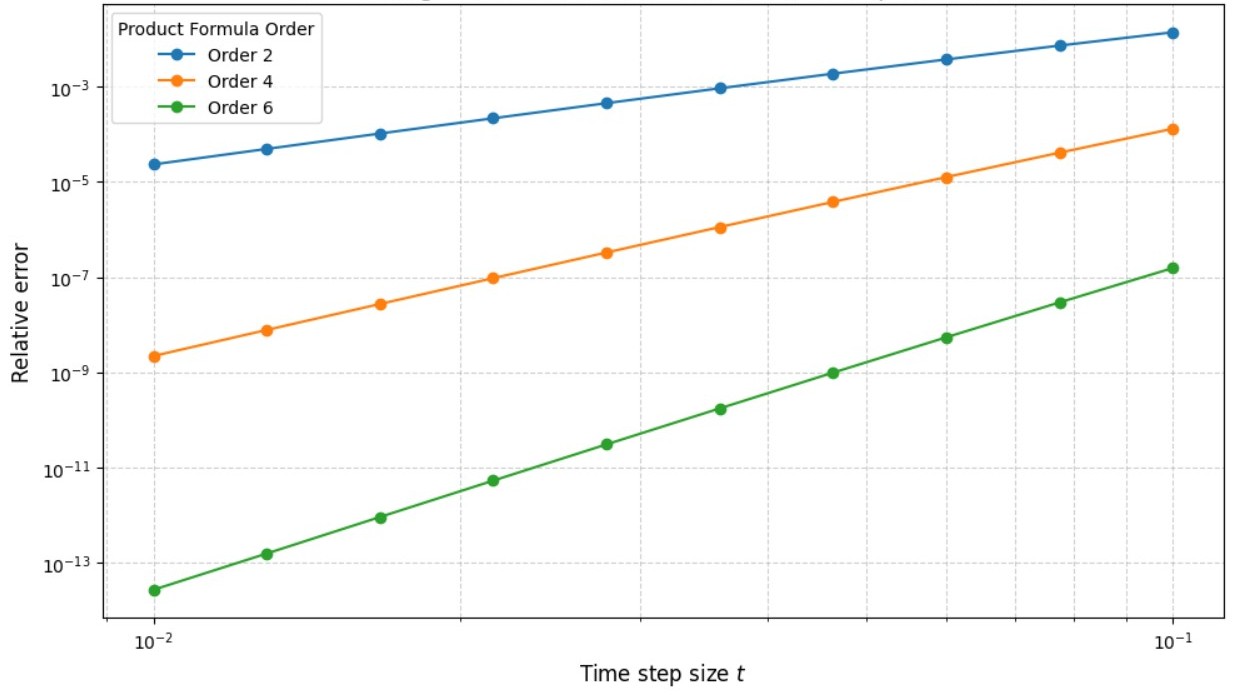}
    \caption{Error scaling of Suzuki product formulas with respect to time}
    \label{fig:error-vs-time}
\end{figure} This plot verifies the expected convergence behavior: since the error scales as \( \mathcal{O}(t^{2k+1}) \), the log-log plot of error versus time exhibits a slope approximately equal to \( 2k+1 \), matching the theoretical prediction.
\\\\
In the second experiment, we assess the accuracy and efficiency of Suzuki product formulas of orders \( 2 \), \( 4 \), \( 6 \) and \( 8 \) in approximating \( e^A \) over a fixed time interval of length 1. The interval is subdivided into \( m \) equal steps, where \( m \) ranges over a set of 20 logarithmically spaced values tailored to each order such that, at the largest \( m \), the relative error is approximately the same across all orders. The computational cost is measured by the total number of matrix exponential evaluations, which equals \( m \cdot (2 \cdot 5^{k-1} + 1) \), since each application of \( S_{2k} \) involves \( 2 \cdot 5^{k-1} + 1 \) exponentials (with \( L = 2 \) in this case). The relative spectral norm error is then plotted against the computational cost on a logarithmic scale in Figure \ref{fig:error-vs-cost}.
\\\\
\begin{figure}[H]
    \centering
     \includegraphics[width=\textwidth]{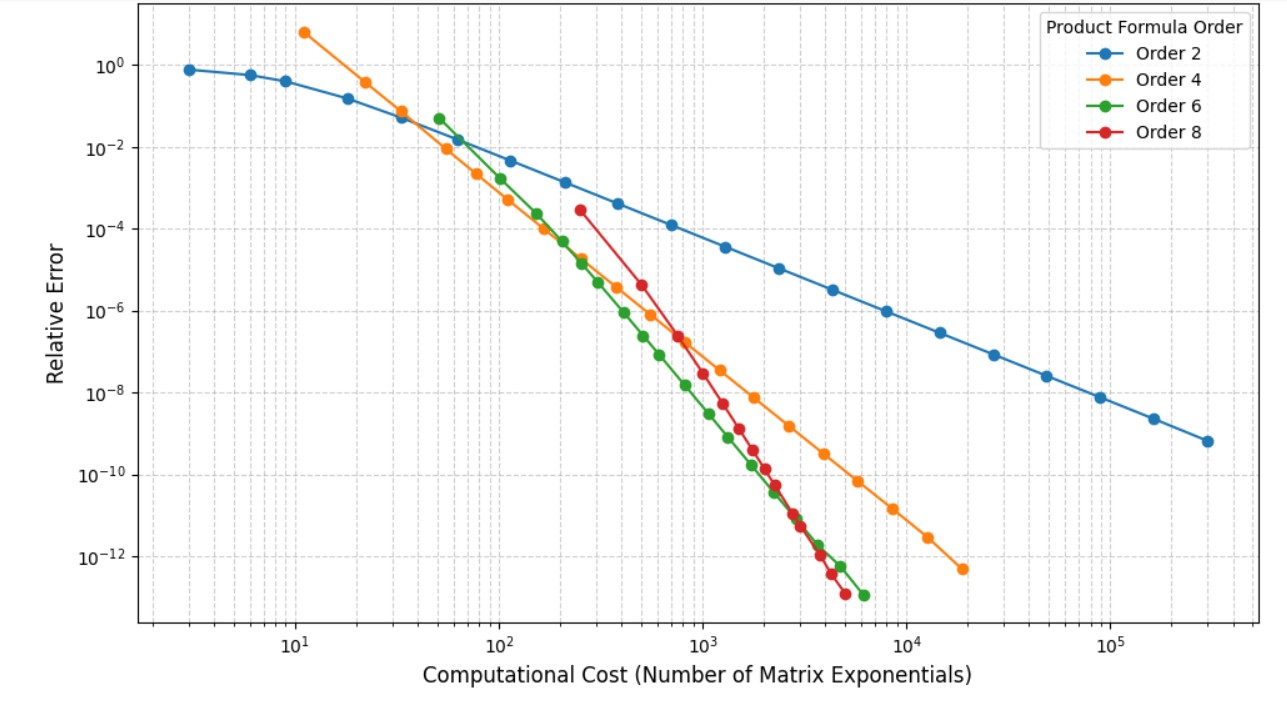}
    \caption{Error vs computational cost for Suzuki product formulas}
    \label{fig:error-vs-cost}
\end{figure} 
The plot demonstrates that, for sufficiently small error tolerance \( \epsilon \), higher-order formulas achieve said precision at reduced computational cost. This confirms the theoretical advantage of higher-order Suzuki product formulas.
\\\\
To evaluate the tightness of the theoretical bound provided by Theorem \ref{nexp}, we compare the predicted and observed number of time steps \( m \) required to achieve a prescribed error tolerance \( \epsilon \in [10^{-6}, 10^{-3}] \). Specifically, we consider ten values of \( \epsilon \), logarithmically spaced between \( 10^{-3} \) and \( 10^{-6} \). For each value of \( \epsilon \), we compute two quantities:
\\\\
First, the theoretical step count \( m_{\text{theory}} \) is derived from Theorem~\ref{nexp} and is given by
\[
m_{\text{theory}} = \left\lceil \frac{(2L5^{k-1}\tau)^{1 + \frac{1}{2k}}}{\epsilon^{1/2k}} \right\rceil,
\]
where \( k = 2 \) for the fourth-order product formula and $L=2$ for this example.
\\\\
Second, the empirical step count \( m_{\text{empirical}} \) is determined by computing the approximation \( S_{2k}(1/m)^m \) for increasing values of \( m \), until the absolute spectral norm error with respect to the exact exponential \( e^A \) falls below the desired tolerance \( \epsilon \).
\\\\
The resulting log-log plot in Figure \ref{fig:theory-vs-empirical} compares \( m_{\text{theory}} \) and \( m_{\text{empirical}} \) as functions of \( \epsilon \).

\begin{figure}[H]
    \centering
    \includegraphics[width=\textwidth]{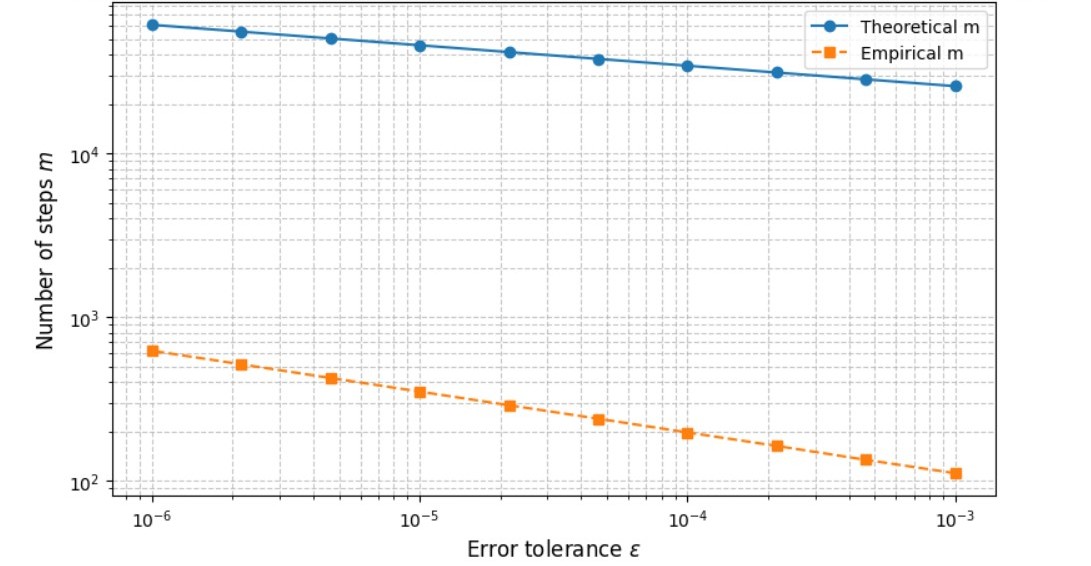}
    \caption{Comparison of theoretical and empirical number of time steps \( m \) needed to achieve error tolerance \( \epsilon \) using the fourth-order Suzuki product formula}
    \label{fig:theory-vs-empirical}
\end{figure}
As shown in the figure, the theoretical estimate from Theorem~\ref{nexp} consistently overestimates the number of steps required by several orders of magnitude. This suggests that the bound is loose in practice. Developing tighter and more realistic complexity bounds for Suzuki's product formulas remains an open and active area of research as we will discuss in the next chapter.
\pagebreak

\chapter{Advances and Open Problems in Hamiltonian Simulation}

In this chapter, we provide a concise overview of the current state of Hamiltonian simulation, focusing primarily on product formula methods. Subsequently, we discuss alternative approaches and highlight some of the active research directions in this field, to the best of the author's knowledge.
\\\\
In the previous section, we saw that our bound for the number of steps \( m \) required to achieve precision \(\epsilon\), discussed in the analysis in \cite{berry2007efficient}, was very loose (off by several orders of magnitude). This looseness also extends to our bounds on the number of exponentials \( N_{\text{exp}} \) and, consequently, the number of quantum gates required to simulate a certain system using product formulas. Childs et al. \cite{childs2018toward} identified this as a major challenge in using product formulas. Although subsequent work has refined these bounds for specific product formulas, such as Suzuki's of a fixed order, or under additional assumptions about the Hamiltonian, it was only recently that a general tight bound was established for systems with nearest-neighbor and power-law interactions, applicable to arbitrary product formulas \cite{childs2021theory}. While the derivation of this Trotter error theory is considerably more involved, the resulting error bounds are concise and computationally efficient. Nevertheless, the tightness of these bounds for general systems remains an open question.
\\\\
Another intriguing aspect is the exploration of product formulas beyond Suzuki's construction. Suzuki's method provides a systematic way to obtain product formulas of arbitrary order; however, the number of exponentials grows exponentially with the order. This naturally raises the question of whether more efficient families of product formulas exist. One notable example is Yoshida's method, which can yield product formulas with fewer exponentials. Similar to Suzuki's formulas, these are constructed as products of \( S_2(t) \) for various time intervals, but unlike the fractal approach, there is no explicit analytic form for the higher-order formulas. Instead, one must solve a complex system of nonlinear polynomial equations. A detailed discussion of Yoshida's product formulas and other alternatives in Hamiltonian simulation is available in \cite{morales2025selection}. Notably, product formulas for Hamiltonian simulation closely resemble classical splitting methods for differential equations. Thus, techniques from classical numerical analysis could potentially be adapted for quantum simulation. An overview of such methods is provided in \cite{blanes2024splitting}.
\\\\
Product formulas also enable simulation of a broader class of Hamiltonians beyond \( k \)-local ones. In particular, it has been shown that sparse Hamiltonians are efficiently simulatable \cite{berry2007efficient}. The key idea is to decompose \( H \) into a sum of polynomially many Hamiltonians \( H_m \), each of which is 1-sparse and then prove that 1-sparse Hamiltonians can be efficiently simulated. A product formula is then used to simulate the full Hamiltonian. Apart from the sparsity assumption, it is necessary to assume the Hamiltonian is row-computable, meaning that one can efficiently determine the positions and values of non-zero entries, often represented using an oracle model. This essentially encodes the structured nature of the Hamiltonian, since its size grows exponentially with the system size \( n \), making direct scanning impractical. The current literature typically works directly with sparse Hamiltonians under the oracle model.
\\\\
With Suzuki's product formulas, we can achieve arbitrarily close to linear scaling in \( t \) and \( \frac{1}{\epsilon} \). It is natural to ask whether this scaling can be further improved. Lower bounds are well known: linear scaling in \( t \) is optimal \cite{berry2007efficient} and polylogarithmic scaling in \( \frac{1}{\epsilon} \) is optimal \cite{berry2014exponential}. However, product formulas have not yet reached these optimal scalings. To achieve this, one must consider methods beyond product formulas. Namely, proper quantum algorithms that typically require significant background to understand. For instance, Berry et al. \cite{berry2015simulating} introduced an approach that achieves polylogarithmic scaling in precision by truncating the Taylor series expansion of the evolution operator and implementing this truncation via a quantum algorithm, despite the truncation not being unitary. Later, Low and Chuang \cite{low2017optimal} developed an algorithm for Hamiltonian simulation that is optimal with respect to all parameters. Despite their theoretical advantages, product formulas remain competitive in practice due to their typically low empirical error and simplicity. Often, the required precision does not necessitate polylogarithmic scaling, diminishing this theoretical advantage. Moreover, product formulas require no ancillary qubits; simulating a system of \( n \) particles straightforwardly requires only \( n \) qubits. In contrast, more advanced algorithms often demand auxiliary qubits, complex state preparation and classical post-processing, which can impact practical performance. A comparative discussion of these algorithms in practice is provided in \cite{childs2018toward}.
\\\\
Finally, it is important to note that this work has focused exclusively on time-independent Hamiltonians. The time-dependent case is generally much more challenging. The analytical solution of the time-dependent Schr\"{o}dinger equation is non-trivial and is typically treated in advanced quantum mechanics texts. While progress has been made, the repertoire of approaches for time-dependent Hamiltonian simulation remains relatively limited compared to the time-independent case. A treatment of product formulas for time-dependent Hamiltonians can be found in \cite{wiebe2011simulating}.

\newpage

\appendix
\chapter{Numerical Validation Code}
\begin{lstlisting}[style=custompython, caption={Common imports and matrix definitions}]
import math
import numpy as np
from numpy.linalg import matrix_power
from scipy.linalg import expm, norm
from scipy.stats import linregress
import matplotlib.pyplot as plt


A = np.array([
    [2.2, 6.9],
    [4.20, 6.66]
])

print("Matrix A:\n", A)

exp_A = expm(A)
print("\nMatrix exponential e^A:\n", exp_A)

half_diag = 0.5 * np.diag(np.diag(A))

B = np.triu(A, k=1) + half_diag

C = np.tril(A, k=-1) + half_diag

print("\nMatrix B (upper triangular):\n", B)
print("\nMatrix C (lower triangular):\n", C)

commutator = B @ C - C @ B

print("\nCommutator [B, C]:\n", commutator)
print("\nNorm of the commutator:\n", norm(commutator, 2))

exp_B = expm(B)
exp_C = expm(C)

print("\nMatrix exponential e^B:\n", exp_B)
print("\nMatrix exponential e^C:\n", exp_C)
\end{lstlisting}
\newpage
\begin{lstlisting}[style=custompython, caption={Functions defining the second- and higher-order Suzuki product formulas}]
def s(k):
    """
    Compute the Suzuki coefficient s_k for the 2k-th order decomposition.

    Parameters:
    - k (int): Order of the decomposition (must be even and >= 2)

    Returns:
    - float: Suzuki coefficient s_k
    """
    if k <= 2 or k % 2 != 0:
        raise ValueError("k must be greater than 2 and even for s_k to be meaningful")

    # Compute s_k = 1 / (4 - 4^{1 / (k - 1)})
    root = 4 ** (1 / (k - 1))
    denominator = 4 - root
    return 1 / denominator

def S2(B, C):
    """
    Compute the second-order Suzuki-Trotter formula.

    Approximates exp(B + C) ~= exp(B/2) · exp(C) · exp(B/2)

    Parameters:
    - B, C (np.ndarray): Matrices to exponentiate

    Returns:
    - np.ndarray: Approximation to exp(B + C)
    """
    return expm(B / 2.0) @ expm(C) @ expm(B / 2.0)

def S2k(B, C, k):
    """
    Compute the 2k-th order Suzuki-Trotter decomposition recursively.

    Parameters:
    - B, C (np.ndarray): Matrices to exponentiate
    - k (int): Desired even order of the decomposition (k must be even, >= 2)

    Returns:
    - np.ndarray: Higher-order approximation to exp(B + C)
    """
    if k % 2 != 0 or k < 2:
        raise ValueError("k must be an even integer >= 2")
    if k == 2:
        return S2(B, C)

    sk = s(k)
    # Recursive structure of the 2k-th order decomposition
    term = S2k(sk * B, sk * C, k - 2)
    middle = S2k((1 - 4 * sk) * B, (1 - 4 * sk) * C, k - 2)
    return term @ term @ middle @ term @ term

def num_exp(k, m):
    """
    Returns the number of matrix exponentials required
    for a 2k-th order Suzuki decomposition over m time steps.

    Parameters:
    - k (int): Order of the decomposition (must be even)
    - m (int): Number of steps in time discretization

    Returns:
    - int: Estimated number of matrix exponentials
    """
    if k % 2 != 0 or k < 2:
        raise ValueError("k must be an even integer >= 2")

    return m * (2 * 5**(k // 2 - 1) + 1)
\end{lstlisting}
\newpage
\begin{lstlisting}[style=custompython, caption={Generates Figure 3.3: Error Scaling with respect to Time}]
orders = [2, 4, 6]

t_values = np.logspace(-2, -1, 10)

errors = {k: [] for k in orders}

for k in orders:
    for t in t_values:
        eAt = expm(A * t)

        approx = S2k(B * t, C * t, k)

        error = norm(approx - eAt, 2) / norm(eAt, 2)
        errors[k].append(error)

# --- Plot: Error vs Time Step Size (log-log) ---

plt.figure(figsize=(10, 6))

for k in orders:
    plt.loglog(t_values, errors[k], 'o-', label=f'Order {k}')

plt.xlabel('Time step size $t$', fontsize=12)
plt.ylabel('Relative error', fontsize=12)
plt.title('Error scaling of Suzuki Product Formulas with respect to time t', fontsize=14)
plt.legend(title='Product Formula Order')
plt.grid(True, which='both', linestyle='--', alpha=0.6)
plt.tight_layout()
plt.show()

# --- Log-log Linear Regression to Estimate Slope ---

log_t = np.log(t_values)

print("Estimated convergence rate (slope) for each order:\n")
for k in orders:
    log_err = np.log(errors[k])
    slope, intercept, r_value, _, _ = linregress(log_t, log_err)

    print(f"  Order {k}: slope ~= {slope:.4f}, R^2 = {r_value**2:.4f}")
\end{lstlisting}
\newpage
\begin{lstlisting}[style=custompython, caption={Generates Figure 3.4: Error vs Computational Cost}]
def get_m_values(max_m, num_points=20):
    # Generate log-spaced points between 1 and max_m
    ms = np.logspace(0, np.log10(max_m), num=num_points)
    ms_rounded = np.unique(np.round(ms).astype(int))  # round and remove duplicates
    return ms_rounded

orders = [2, 4, 6, 8]

# Define max_m for each order
m_settings = {
    2: 100000,  
    4: 1700,
    6: 120,
    8: 20,
}

errors = {k: [] for k in orders}
costs = {k: [] for k in orders}

for k in orders:
    max_m = m_settings[k]
    m_values = get_m_values(max_m)
    for m in m_values:
        step_approx = S2k(B / m, C / m, k)
        approx = matrix_power(step_approx, m)

        error = norm(approx - exp_A, 2) / norm(exp_A, 2)
        errors[k].append(error)

        cost = num_exp(k, m)
        costs[k].append(cost)

# --- Plotting Error vs Computational Cost ---

plt.figure(figsize=(10, 6))

for k in orders:
    plt.loglog(costs[k], errors[k], 'o-', label=f'Order {k}')

plt.xlabel('Computational Cost (Number of Matrix Exponentials)', fontsize=12)
plt.ylabel('Relative Error', fontsize=12)
plt.title('Error vs Computational Cost for Suzuki Product Formulas', fontsize=14)
plt.legend(title="Product Formula Order")
plt.grid(True, which="both", linestyle="--", alpha=0.6)
plt.tight_layout()
plt.show()
\end{lstlisting}
\newpage
\begin{lstlisting}[style=custompython, caption={Generates Figure 3.5: Empirical vs Theoretical Step Count}]
order = 4
epsilons = np.logspace(-3, -6, 10)

m_theory_list = []
m_empirical_list = []

L = 2
tau = max(norm(B, 2), norm(C, 2))
for eps in epsilons:
    # --- Theoretical minimum number of steps m to achieve error <= eps ---
    factor = (2 * L * 5**(order - 1) * tau)**(1 + 1 / (2 * order))
    m_theory = math.ceil(factor / (eps**(1 / (2 * order))))
    m_theory_list.append(m_theory)

    # --- Empirical determination of minimal r ---
    for m in range(1, 1000):
        s2k_step = S2k(B / m, C / m, order)
        approx = matrix_power(s2k_step, m)

        err = norm(approx - exp_A, 2)
        if err <= eps:
            m_empirical_list.append(m)
            break
    else:
        m_empirical_list.append(np.nan)

# --- Plot results ---
plt.figure(figsize=(8, 5))
plt.loglog(epsilons, m_theory_list, 'o-', label='Theoretical m')
plt.loglog(epsilons, m_empirical_list, 's--', label='Empirical m')
plt.xlabel('Error tolerance $\\epsilon$', fontsize=12)
plt.ylabel('Number of steps $m$', fontsize=12)
plt.title('Comparison of theoretical and empirical step counts for 4th-order Suzuki formula', fontsize=14)
plt.legend()
plt.grid(True, which='both', linestyle='--', alpha=0.7)
plt.tight_layout()
plt.show()
\end{lstlisting}
\pagebreak

\addcontentsline{toc}{chapter}{Bibliography}
\bibliographystyle{plain}
\nocite{*}
\bibliography{main}

\newpage

\end{document}